\documentclass[letter,11pt]{article}
\usepackage[margin=1in]{geometry}
\usepackage[numbers,sort]{natbib}
\usepackage{authblk}
\usepackage[T1]{fontenc}
\usepackage{amsmath,amsthm,amssymb,thm-restate}
\usepackage{algorithm}
\usepackage[noend]{algpseudocode}
\usepackage[colorlinks=true,citecolor=blue,urlcolor=blue]{hyperref}
\usepackage{amsfonts}
\usepackage{dsfont}
\usepackage{tikz}
\usepackage[small,bf,hypcap=true]{caption}
\usepackage{comment}
\usepackage[shortlabels]{enumitem}
\usepackage{bm}

\usepackage{cleveref}
\usepackage{xspace}
\usepackage{xcolor}
\usepackage[colorinlistoftodos,textsize=tiny,textwidth=2cm,color=red!25!white,obeyFinal]{todonotes}

\theoremstyle{plain}

\newtheorem{thm}{Theorem}[section]

\newtheorem*{thm*}{Theorem}

\newtheorem{conj}[thm]{Conjecture}
\newtheorem{fact}[thm]{Fact}
\newtheorem{lem}[thm]{Lemma}
\newtheorem{Def}[thm]{Definition}
\newtheorem{obs}[thm]{Observation}
\newtheorem{claim}[thm]{Claim}

\newcommand{\E}{\mathbb{E}\xspace}

\newcommand{\D}{\Delta}

\newcommand{\eat}[1]{}
\newcommand{\generalLB}{1.606}
\newcommand{\generalUB}{1.777}
\newcommand{\eps}{\epsilon}

\newcommand{\biUBBeta}{\beta\log \,\frac{\beta}{\beta - 1}}
\newcommand{\generalUBBeta}{\beta^2 - \beta + \beta\log \frac{1}{\beta - 1}}
\newcommand{\T}{\mathcal{S}}

\newcommand{\iteration}{iteration\xspace}
\newcommand{\iterations}{iterations\xspace}
\newcommand{\phase}{phase\xspace}
\newcommand{\phases}{phases\xspace}
\newcommand{\marking}{\textsc{marking}\xspace}

\newcommand{\rootconstant}{12}
\newcommand{\multiplicativeconstant}{24\cdot 10^{\rootconstant}}
\newcommand{\numphases}{\lceil (4/p)\log (1/p) \rceil}

\title{Tight Bounds for Online Edge Coloring}

\author[1]{Ilan Reuven Cohen\thanks{Email address: \url{ilanrcohen@gmail.com}.}}
\author[2]{Binghui Peng\thanks{Email address: \url{pbh15@mails.tsinghua.edu.cn}.}\thanks{Work done in part while the author was visiting Carnegie Mellon University.}}
\author[3]{David Wajc\thanks{Email address: \url{dwajc@cs.cmu.edu}.}\thanks{Work done in part while the author was visiting EPFL.}\thanks{Supported in part by NSF grants CCF-1618280, CCF-1814603, CCF-1527110, NSF CAREER award CCF-1750808 and Sloan Research Fellowship.}}
\affil[1]{Carnegie Mellon University and University of Pittsburgh}
\affil[2]{Tsinghua University}
\affil[3]{Carnegie Mellon University}

\date{\vspace{-1cm}}

\begin{document}

\maketitle
	\pagenumbering{gobble}

\begin{abstract}
	Vizing's celebrated theorem asserts that any graph of maximum degree $\Delta$ admits an edge coloring using at most $\Delta+1$ colors.
	In contrast, Bar-Noy, Naor and Motwani showed over a quarter century that the trivial greedy algorithm, which uses $2\Delta-1$ colors, is optimal among online algorithms. Their lower bound has a caveat, however: it only applies to low-degree graphs, with $\Delta=O(\log n)$, and they conjectured the existence of online algorithms using $\Delta(1+o(1))$ colors for $\Delta=\omega(\log n)$. 
	Progress towards resolving this conjecture was only made under stochastic arrivals (Aggarwal et al., FOCS'03 and Bahmani et al., SODA'10).

	\medskip

	We resolve the above conjecture for \emph{adversarial} vertex arrivals in bipartite graphs, for which we present a $(1+o(1))\Delta$-edge-coloring algorithm for $\Delta=\omega(\log n)$ known a priori. Surprisingly, if $\Delta$ is not known ahead of time, we show that no $\big(\frac{e}{e-1} - \Omega(1) \big) \Delta$-edge-coloring algorithm exists. We then provide an optimal, $\big(\frac{e}{e-1}+o(1)\big)\Delta$-edge-coloring algorithm for unknown $\Delta=\omega(\log n)$. 
	Key to our results, and of possible independent interest, is a novel fractional relaxation for edge coloring, for which we present optimal fractional online algorithms and a near-lossless online rounding scheme, yielding our optimal randomized algorithms.
\end{abstract}

\newpage
\pagenumbering{arabic}		
	\section{Introduction}

Edge coloring is the problem of assigning a color to each edge of a multigraph so that no two edges with a common endpoint have the same color.
This classic problem, even restricted to bipartite graphs, can be used to model scheduling problems arising in sensor networks \cite{gandham2008link}, switch routing \cite{aggarwal2003switch}, radio-hop networks \cite{tassiulas1992stability} and optical networks \cite{rasala2005strictly}, among others.
Edge coloring can trace its origins back to the 19th-century works of 
\citet{tait1880remarks} and \citet{petersen1898theoreme}, who studied this problem in the context of the four color theorem. 
Shannon \cite{shannon1949theorem} later studied edge coloring in the context of color coding wires in electrical units, and proved that any multigraph $G$ of maximum degree $\Delta=\Delta(G)$ admits a $\lfloor \frac{3\Delta}{2}\rfloor$-edge-coloring; i.e., a coloring using at most $\lfloor \frac{3\Delta}{2}\rfloor$ colors. (This is tight.)
Inspired by this result, Vizing \cite{vizing1964estimate} proved that any simple graph can be edge colored using $\Delta+1$ colors. Clearly, $\Delta$ colors are necessary to edge color a graph, and for bipartite graphs 
multiple near-linear-time $\Delta$-edge-coloring algorithms are known~\cite{cole2001edge,alon2003simple,goel2013perfect}. For general graphs, several polytime $(\Delta+1)$-edge-coloring algorithms are known \cite{gabow1985algorithms,misra1992constructive,vizing1964estimate}, and this too is likely optimal, as determining whether a general graph is $\Delta$-edge-colorable is NP-hard \cite{holyer1981np}.

In addition to these optimal polytime algorithms, there exists a simple quasilinear-time $(2\Delta-1)$-edge-coloring greedy algorithm, which colors each edge with the lowest color unused by its adjacent edges. 
The greedy algorithm is implementable in many restricted models of computation, and improving upon its coloring guarantees, or even matching them quickly in such models, has been the subject of intense research. Examples include PRAM \cite{lev1981fast}, NC and RNC \cite{karloff1987efficient,berger1991simulating,motwani1994probabilistic}, dynamic \cite{bhattacharya2018dynamic,duan2019dynamic} and distributed algorithms (e.g.,  \cite{panconesi1997randomized,dubhashi1998near,ghaffari2018deterministic,fischer2017deterministic,chang2018optimal,elkin20152delta}).

For online algorithms, little progress was made towards beating the greedy algorithm. The only positive results are 
under \emph{random-order} edge arrival in ``dense'' bipartite multigraphs.
Specifically, under such stochastic arrivals, \citet{aggarwal2003switch} showed how to obtain a $\Delta(1+o(1))$ edge coloring of $n$-vertex multigraphs if $\Delta=\omega(n^2)$ and \citet{bahmani2012online} showed how to obtain a $1.26\Delta$ edge coloring under the milder assumption that $\Delta=\omega(\log n)$.
The lack of progress for adversarial arrival order is likely explained by the following theorem of Bar-Noy et al.~\cite{bar1992greedy}.

\begin{thm}[\cite{bar1992greedy}, informal]
	No online edge coloring algorithm can $2\Delta-2$ edge-color a graph.
\end{thm}

However, the lower bound of Bar-Noy et al.~requires a number of nodes $n$ \emph{exponential} in $\Delta$: that is, it only holds for some $\Delta=O(\log n)$.
Therefore, this lower bound can be thought of as an \emph{additive} lower bound of $\Delta + \Omega(\log n)$, rather than a multiplicative lower bound of $\approx 2\Delta$.
Put otherwise, it does not preclude a better-than-$(2\Delta-1)$ edge-coloring algorithm for $\Delta = \Omega(\log n)$ large enough. 
Indeed, Bar-Noy et al.~went so far as to conjecture that $\Delta(1+o(1))$-edge-colorings are computable online for large enough $\Delta$.
\begin{conj}[\cite{bar1992greedy}]\label{conj:bnm}
	There exists an online algorithm which $\Delta(1+o(1))$ edge-colors graphs with $\Delta = \omega(\log n)$.
\end{conj}

\vspace{-0.49cm}
\paragraph{Our focus.} In this paper we study edge coloring under the adversarial online vertex-arrival model of \citet{karp1990optimal}, where vertices on one side of a bipartite graph arrive over time, with their edges to previously-arrived neighbors. In this model, an online algorithm colors each edge $e$ upon arrival, immediately and irrevocably.
Recall from above that edge coloring in bipartite graphs 
has multiple applications, including in online settings.
Indeed, such an application of edge coloring bipartite graphs to switch routing (with input switches on one side and output switches on the other) was precisely the motivation of \citet{aggarwal2003switch} to study online edge coloring.
The online edge coloring lower bound of \cite{bar1992greedy}~holds even for bipartite vertex arrivals. We show that for such arrivals a large enough maximum degree indeed allows to circumvent Bar-Noy et al.'s lower bound, and prove their 
conjecture. In particular, we present optimal algorithms (up to $o(1)$ terms) both for the known-$\Delta$ scenario, as well as for the stricter online problem where $\Delta(=OPT)$ is unknown a priori.

\subsection{Our Contributions}

We provide the following optimal results for online edge coloring under adversarial vertex arrivals. For conciseness, we state our results in terms of competitiveness, calling an $\alpha\cdot\Delta$-edge-coloring algorithm \emph{$\alpha$-competitive}, as the optimal edge coloring requires at least $\Delta$ colors.

Our first result is an optimal algorithm for known large $\Delta$.

\begin{thm}\label{thm:known-delta}
	There exists a $(1+o(1))$-competitive randomized edge coloring algorithm for bipartite graphs of known maximum degree $\Delta = \omega(\log n)$. A competitive ratio of $1+o(1)$ is optimal -- no randomized algorithm is  $1+o(1/\sqrt\Delta)$ competitive for any $\Delta$.
\end{thm}

Like all prior non-trivial online algorithms, the above algorithm assumes a priori knowledge of a critical parameter of the input, namely $\Delta$, which is the optimum number of colors needed to color the bipartite graph. 
However, in many online scenarios, such assumptions are unreasonable.
We show that removing this assumption results in a strictly harder problem, though here too greedy is suboptimal. Our main contribution is an optimal online algorithm for unknown large $\Delta$.

\begin{thm}\label{thm:unknown-delta}
	There exists an $(\frac{e}{e-1}+o(1))$-competitive randomized edge coloring algorithm for bipartite graphs of unknown maximum degree $\Delta = \omega(\log n)$. 
	This is optimal (up to $o(1)$ terms) -- no algorithm is better than $\frac{e}{e-1}$ competitive for unknown $\Delta$.
\end{thm}

\noindent\textbf{Remark 1.} For simplicity we stated our positive results in theorems \ref{thm:known-delta} and \ref{thm:unknown-delta} for $\Delta=\omega(\log n)$. More generally, our algorithms' competitive ratios are of the form $\alpha+O(\sqrt[c]{\log n/\Delta}))$ for some constant $c\geq 1$ and $\alpha=1$ or $\alpha=\frac{e}{e-1}$, respectively (see \Cref{sec:online-rounding}). Thus, we obtain better-than-$2$ competitive ratios already for sufficiently large $\Delta=O(\log n)$.

\noindent\textbf{Remark 2.} We stated all our positive results for simple graphs, though they hold more generally for any multigraph with maximum edge multiplicity $o(\Delta)$. (A necessary condition -- see \Cref{sec:multigraphs}).

Our upper and lower bounds rely on a new fractional relaxation for edge coloring, of possible independent interest. In particular, we present matching upper and lower bounds for this relaxation, and present a nearly-lossless online rounding of solutions to this relaxation. Using this relaxation, we show a complementary result: a separation between online edge coloring on general and bipartite graphs, which we prove by showing a higher lower bound for the former (as well as a better-than-greedy fractional algorithm for the latter).

\begin{thm}\label{thm:general-graphs}
	No fractional online edge coloring algorithm is better than $1.606  (> \frac{e}{e-1})$ competitive for general graphs of unknown maximum degree $\Delta$. On the other hand, there exists a fractional edge coloring algorithm which is $1.777$ competitive for general graphs of unknown maximum degree $\Delta$.
\end{thm}

To conclude, relying on our new relaxation, we present the first online algorithms beating greedy under any adversarial arrivals. In particular, we 
prove the conjecture of Bar-Noy et al.~and 
provide tight bounds for the well-studied model of Karp et al., both under known and unknown $\Delta$.

\subsection{Techniques}\label{sec:tec}   

\paragraph{Novel Relaxation.} 
The classic fractional relaxation for edge coloring asks to minimize $\sum_M x_M$, subject to $\sum_{M\ni e} x_e = 1$ for 
every $e\in E$ and $x_M\geq 0$ for every matching $M$. That is, this relaxation fractionally uses integral matchings to cover each edge.
For the online problem, this standard relaxation is not particularly useful, as the set of matchings in the input 
graph is unknown a priori (since the edge set is unknown). Our first insight is a novel fractional 
relaxation which allows for more ``myopic'' assignments upon vertex arrivals -- its variables $x_{e,c}$ are the extent to which an edge $e$ is colored $c$. 
Specifically, our relaxation \emph{integrally} uses \emph{fractional} matchings to cover each edge; i.e., $\sum_{c} x_{e,c}=1$. 
The goal is to minimize the number of non-zero fractional matchings used.
This simple change to the 
relaxation proves particularly useful in the online setting, and underlies both our upper and lower bounds. We believe this 
relaxation may find applications to other incomplete-information models, such as dynamic, distributed, and local computation algorithms.

\paragraph{Fractional Algorithms.} 
Our relaxation admits a trivial $1$-competitive solution: set $x_{e,c} = 1/\Delta$, for each edge and $c\in [\Delta]$. If $\Delta$ is known a priori, this solution can even be computed by an \emph{online} fractional algorithm (which must fix the values $x_{ec}$ for each edge $e$ immediately on arrival).
By our lower bounds, for unknown $\Delta$, other algorithms are needed. A natural candidate is the 
greedy ``water-filling'' algorithm, which continuously increases an arriving edge's assignment for all colors minimizing the maximum load for either endpoint. 
However, this approach may yield an extremely unbalanced allocation.
In particular, it can make a vertex $v$ have load of one in half of the colors used so far and another vertex $u$ have load of one in the other half. Adding an edge $(u,v)$ would then force the algorithm to open a new color -- resulting in the trivial $2\Delta-1$ bound. (See \Cref{sec:bad-example-wf}.) 
Guided by our lower bounds for fractional algorithms, we derive simple but crucial changes to the water-filling algorithm. First, we use an \emph{asymmetric} approach, where we pick colors to use based only on the load on the offline vertex. Secondly, in order to bound the load on the online vertex, we cap the value of each edge-color pair. These changes yield  bounded loads for the online vertex and a more balanced allocation, resulting in an \emph{optimal} online fractional edge coloring.

\paragraph{Online Rounding.} 
Given an $\alpha$-competitive fractional online algorithm, our approach would be to use its $\alpha\Delta$ fractional matchings 
to obtain $\alpha\Delta$ \emph{integral} matchings, or colors, which leave the remaining uncolored subgraph having low 
maximum-degree ($o(\Delta)$). Then using greedy on the remaining uncolored edges requires a further $o(\Delta)$ 
colors, or $(\alpha+o(1))\Delta$ colors overall. The question is how to compute $\alpha\Delta$ colors based on the online 
fractional edge coloring. One natural way to do so is to repeatedly round these fractional matchings online, using a 
near-lossless online rounding scheme for fractional matchings (\cite{cohen2018randomized}). Unfortunately, while maximum-degree vertices stand to be 
matched $\approx \Delta$ times during such a rounding stage, a constant fraction of these matches would be along 
previously-matched (colored) edges. To see this, note that each edge has a constant probability of being matched at least 
twice this way (since each edge has a constant probability of being colored more than once).

We therefore employ a more elaborate approach, repeatedly rounding \emph{subsets} of \emph{multiple} fractional edge 
colorings' fractional matchings. Our guiding intuition is the following simple observation, that the load assigned to the 
edges of a maximum-degree vertex by an average fractional matching among the $\alpha\Delta$ matchings is precisely 
$1/\alpha$. Consequently, rounding a randomly-chosen fractional matching will result in this vertex being matched 
with probability $\approx 1/\alpha$. If most such matches are along previously-unmatched(uncolored) edges, then this 
vertex's degree in the uncolored graph will decrease at the appropriate rate. However, as exhibited by the previous 
approach, the probability of matching along an uncolored edge decreases when we use round many fractional matchings. 
Therefore, we only round a \emph{subset} of fractional matchings, small enough to not decrease the probability of $v$ being matched along an uncolored edge (due to re-matches), yet large enough to apply tail bound and argue that all high-degree 
vertices have their uncolored degree decrease at a rate of $\approx 1/\alpha$ per color used. 

How to sample such a number of fractional matchings, \emph{not too few and not too many, but just right}, for unknown $\Delta$ is not immediate, however, as we do not even know how many fractional matchings we will use (as $\Delta$ is unknown and keeps increasing). To address this, we rely on the fact that $\Delta = \omega(\log n)$, allowing us to argue that sampling each fractional matching (including matchings which are currently trivial) a priori with appropriate probability $p$ gives the following. We choose $p=o(1)$ (guaranteeing few re-colors) which also satisfies $\Delta \cdot p = \omega(\log n)$ (in order to have concentration up to $(1\pm o(1))$ factors on the number of non-trivial colors used), we color a $\approx 1-p$ fraction of the edges of high-degree vertices (i.e., $\approx \Delta \cdot p)$ edges, while using only $\approx \alpha \Delta \cdot p$ colors, all with high probability.
We then compute \emph{another} fractional edge coloring, but this time only on the residual uncolored subgraph. Repeatedly 
applying this approach (computing another edge colorings, and rounding a sampled subset of its matchings) therefore 
allows us to decrease the maximum degree of the uncolored graph at the required rate w.h.p. So, after $\alpha\Delta$ 
colors are used this way, we safely run greedy, yielding an $(\alpha+o(1))$-competitive edge coloring. Plugging in the 
optimal algorithms for known and unknown $\Delta$ into this rounding scheme then yields our main positive results: \emph{optimal} randomized online edge coloring algorithms.

\paragraph{Lower Bounds.}
For our lower bounds (including the tight ones for bipartite graphs) we formulate linear programs capturing constraints 
satisfied by any $\alpha$-competitive fractional online algorithms (for our relaxation) when run on a tailor-made family of edge coloring instances. 
We then present a family of feasible solutions to the dual program, whose value converges to the claimed lower bounds, implying the lower bound 
on $\alpha$ by LP duality. (See  \cite{azar2017online} for more examples of this approach.)

\subsection{Related Work}\label{sec:related-work}

\paragraph{Online Edge Coloring.} 
Several previous papers studied edge coloring in online settings \cite{bar1992greedy,aggarwal2003switch,bahmani2012online,favrholdt2003line,ehmsen2010comparing,favrholdt2014online,mikkelsen2016randomization,mikkelsen2015optimal}. 
\citet{mikkelsen2016randomization,mikkelsen2015optimal} studied the online edge coloring problem, but with advice about the future.
Favrholdt et al.~\cite{favrholdt2003line,ehmsen2010comparing,favrholdt2014online} studied the ``dual'' problem of maximizing the number of edges colored using a fixed number of colors.
Most relevant to our paper is the work of Motwani et al.~\cite{bar1992greedy,aggarwal2003switch,bahmani2012online}.
\citet{aggarwal2003switch} presented a $(1+o(1))$-competitive algorithm for multigraphs with known $\Delta=\omega(n^2)$.
\citet{bahmani2012online}, inspired by the distributed algorithm of \citet{panconesi1997randomized}, 
gave a $1.26$-competitive algorithm for multigraphs with known $\Delta=\omega(\log n)$. 
Both algorithms require \emph{random order} edge arrivals, 
and fall short of the guarantees of \Cref{conj:bnm} (\cite{bar1992greedy}), either in the competitive ratio or in the requirement of $\Delta$. In contrast, 
we consider vertex arrivals under the stricter \emph{adversarial} arrival order, for which we match these conjectured bounds for known $\Delta$, and also achieve optimal bounds for (harder) unknown $\Delta$.

\paragraph{Online Matching.} As edge coloring is the problem of partitioning a graph's edges into matchings, it is natural that our work relates to the long line of work on online matching. This problem was introduced in the seminal work of \citet{karp1990optimal}, who presented the classic \textsc{ranking} algorithm, which is $(1-\frac{1}{e})$ competitive for bipartite graphs under one-sided arrivals, and proved its optimality. A simpler argument proves this algorithm's optimality even among fractional algorithms (see \cite{feige2018tighter}).
Alternative analyses of this algorithm were given over the years (\cite{goel2008online,birnbaum2008line,devanur2013randomized,eden2018economic}) and another optimal fractional algorithm with further applications, \textsc{water filling}, was given in \cite{kalyanasundaram2000optimal}. Better bounds are known under structural assumptions \cite{buchbinder2007online,naor2018near,cohen2018randomized}, and under stochastic arrivals \cite{feldman2009online,karande2011online,mahdian2011online}.
See the survey of \citet{mehta2013online} for more on this problem and its extensions.
Finally, we note the recent interest in online matching in general graphs \cite{wang2015two,huang2018match,huang2019tight,gamlath2019online}.
Our complementary 
results of \Cref{thm:general-graphs} for online edge coloring in general graphs are another step towards a better understanding of matching-theory-related problems in online models in general graphs.
	\section{A New Fractional Relaxation}\label{sec:fractional}
In this section, we define our new online fractional edge coloring relaxation
and discuss several of its properties.

\paragraph{The Classic Fractional Relaxation.}
The classic relaxation for edge coloring has a nonnegative variable $x_M$ for each
 matching $M$ in $G=(V,E)$, corresponding to the (fractional) extent to which this matching is used in the solution.
The objective is to minimize $\sum_{M} x_M$ subject to $\sum_{M\ni e} x_M = 1$ for each edge $e\in E$.
This relaxation clearly lower bounds the chromatic index; i.e., the minimum number of matchings needed to cover $G$. A long-standing conjecture of Goldberg and Seymour is that this relaxation is at most one lower than the chromatic index \cite{goldberg1973multigraphs,seymour1979multi}. 
(See \cite[Chapter 7.4]{lovasz2009matching} for more discussion of this relaxation.)
Unfortunately, this relaxation seems somewhat unwieldy in an online setting, 
as we outline below.

\subsection{The New Relaxation}\label{sec:fractional-def}

The standard fractional edge coloring relaxation is difficult to use in online settings, 
where we do not know the edges which will arrive in the future, let alone which matchings $G$ will contain.
This motivates us to define a more ``myopic'' relaxation, which allows us to make our (fractional) assignments immediately 
upon an edge's arrival (due to one of its endpoints' arrival). Specifically, rather than relax 
the integrality of the extent to which we use integral matchings, we relax the integrality of the matchings used. 
That is, while the classic relaxation fractionally uses integral matchings to color edges, our 
relaxation \emph{integrally} uses \emph{fractional} matchings to color edges.

The edge coloring relaxation we consider is thus the following. 
We say a graph $G(V,E)$ is \emph{fractionally $k$-edge-colorable} if there is a feasible solution to the linear program
\begin{align*}
\sum_{c \in [k]} x_{e,c} &= 1  \qquad \qquad \forall e \in E\\
\sum_{e\ni v} x_{e,c} &\leq 1 \qquad \qquad \forall v\in V, c \in [k] \\
x_{e,c} &\geq 0 \qquad \qquad  \forall e\in E, c \in [k]
\end{align*}
For any graph $G$, the minimal number of fractional colors $k$ is equal to $G$'s maximum degree,~$\Delta$. We note that in bipartite graphs this relaxation and the classic relaxation are equivalent in an offline sense, in that any solution to one can be transformed to a solution of equal value to the other (for general graphs, there can be a gap of one between the two, as exemplified by the triangle graph). In an online sense it is not clear how to go from one relaxation to the other, and so we will rely only on our new relaxation.

\paragraph{An LP Formulation.} For notational simplicity, rather than 
discuss fractional algorithms using some $k = \alpha \cdot \Delta$ colors, we will instead use 
$k=\Delta$ colors and relax the second constraint to 
\begin{align*}
\sum_{e\ni v} x_{e,c} &\leq \alpha \qquad \qquad \forall v\in V, c \in [\Delta]
\end{align*}
When dealing with fractional solutions, it is easy to ``stretch'' such a solution to obtain a feasible edge coloring (i.e., satisfying $\sum_{e\ni v} x_{e,c} \leq 1$) while using $\lceil\alpha \cdot \Delta\rceil \leq \alpha\cdot \Delta + 1$ colors, and this can be done online.
Therefore, our goal will be to minimize $\alpha$ --- the competitive ratio.

\paragraph{Online Algorithms for the LP Relaxation.} An online fractional edge coloring algorithm must assign $x_{e,c}$ values for all edges $e$ upon arrival, immediately and irrevocably. For example, if $\Delta$ is known a priori, assigning each edge-color pair a value of $\frac{1}{\Delta}$ trivially yields a $1$-competitive online fractional algorithm. 
If $\Delta$ is \emph{unknown}, the situation is not so simple, as our lower bounds of \Cref{sec:hardness} demonstrate. In the following section we present our online fractional algorithms for unknown $\Delta$, including an optimal algorithm for bipartite graphs.
	\section{The Fractional Online Algorithm}\label{sec:fractional-algorithm}

\newcommand{\ktau}{k}
Our LP relaxation asks to minimize the maximum \emph{load} of any vertex $u$ in color $c$, $L_u (c) \triangleq \sum_{e\ni u} x_{e,c}$. The greedy water-filling algorithm, upon arrival of edge $e$, increases all $x_{e,c}$ for all colors $c$ minimizing the maximum load of either endpoint of $e$. This natural algorithm is no better than the integral greedy algorithm, however (see \Cref{sec:bad-example-greedy}). 
In our algorithm, upon arrival of a vertex $v$, we run a variant of the water-filling algorithm on each edge $(u,v)$ in an arbitrary order. One difference in our algorithm compared to the greedy one is that its greedy choice is \emph{asymmetric}, and is only determined by the current loads of the previously-arrived endpoint, $u$.
The second difference is that we set a \emph{bound constraint} of $\beta/\Delta$ for each color per edge, where $\Delta$ is 
the \textbf{current} maximal degree, and $\beta$ is a parameter of the algorithm which will be determined later. The bound constraints result in bounded load trivially for the online vertex, and by careful analysis, also for the offline vertex. In addition, the bound constraints result in a more balanced allocation, which uses more colors for each edge, but fewer colors overall.
A formal description of our algorithm is given in \Cref{alg:cap_water_filling}. Our algorithm is described as a continuous process, but can be discretized easily.

\begin{algorithm}[H] 
	\caption{Bounded Water Filling}
	\label{alg:cap_water_filling}
	
	\begin{algorithmic}[1]
		\Require Online graph $G(V,E)$ with unknown $\Delta(G)$ under vertex arrivals, parameter $\beta\in (1,2)$.\;
		
		\Ensure Fractional edge coloring $\{x_{e,c} \mid e\in E$, $c\in [\Delta(G)]\}$.

		\For{each arrival of a vertex $v$}
		\State $\Delta \leftarrow \max\{d(u)\mid u\in V\}$. \Comment{$\Delta$ = \textbf{current} max.~degree} \label{line:delta-init}
		\For{each $e=(u,v)\in E$}
		\While{$\sum_{c\in [\Delta]}x_{e,c} < 1$} \Comment{initially, $x_{e,c}=0$}
		
		\State let $\mathit{U} \triangleq  \{c \in [\Delta] \mid x_{e,c} < \beta/\Delta\}$. \Comment{``unsaturated'' colors for $e$}
		\State let $\mathit{C} \triangleq \{c \in \mathit{U} \mid L_u(c) = \min_{c\in \mathit{U}} L_u(c)\}$. \Comment{``currently active'' colors for $e$}
		\For{all $c \in \mathit{C}$}
		\State increase $x_{e,c}$ continuously. \Comment{update $L_u(c),L_v(c)$,$\mathit{U}$ and $\mathit{C}$.}
		\EndFor								
		\EndWhile
		\EndFor
		\EndFor
	\end{algorithmic}
\end{algorithm}

\subsection{Basic properties of the algorithm}

Our water filling algorithm preserves important monotonicity properties on the loads of any previously-arrived vertex $v$. In particular, the order obtained by sorting colors by their loads for $v$ remains invariant following its future neighbors' arrivals. 
More formally, for each vertex $v$, we define an order permutation $\sigma_v: Z^+ \rightarrow Z^+$,
where $\sigma_v(i)$ is the index of $v$'s $i^{th}$ most loaded color index after the vertex $v$ arrives and its edges are fractionally colored (e.g., $\sigma_v(1)$ is the most-loaded color index). In addition, we define the load of a color in a vertex with respect to this order; i.e., we denote by $\ell_{u}^{t}(i)$ the load of color $\sigma_u(i)$ for vertex $u$ after its $t^{th}$ neighbor arrives -- which we refer to as \emph{step} $t$. In this notation, our monotonicity property will be that $\ell_{u}^{t}(i)\geq \ell_{u}^{t}(i+1)$ for each $u$ and $i,t\in \mathbb{Z}^+$.

We denote by $\delta_{u}^t$ the global maximum degree after the arrival of the $t^{th}$ neighbor of vertex $u$ and denote by $A_u$ the degree of $u$ when it arrives (e.g., $A_u = 0$ for offline vertices in bipartite graphs). 
Next, we prove properties of the load of a specific vertex $u$ after its arrival (i.e., for steps $t > A_u$), at which point the order $\sigma_u$ is already set. 
For ease of notation we omit the subscript $u$ from  variables $\ell,\delta$ and $A$ whenever it will be clear from context (i.e., when considering a single vertex $u$). 
In addition, as $\sigma$ will be clear from context, we will use \emph{color $k$} as shorthand notation to $\sigma(k)$.
Moreover, due to space constraints, we defer most proofs to \Cref{sec:fractional-deferred}.

We first observe that for our bounded water-filling algorithm (as for its unbounded counterpart), the load of $u$ is monotone decreasing with respect to the $\sigma_u$ order, and for each step $t$, the increase in the load for $i\leq \delta^t$ is monotone increasing in the $\sigma_u$ order.
\begin{obs}\label{obs:monotonicity}
	For all color indices $i$, and for any $t > A$
	\begin{itemize}
		\item $\ell^{t}(i) \geq  \ell^{t}(i+1)$.
		\item $\ell^{t}({i}) - \ell^{t-1}(i) \geq \ell^{t}(i-1) - \ell^{t-1}(i-1) $, for all  $i \leq \delta^{t}$.
	\end{itemize}
\end{obs}

In our analysis, we focus on the \emph{critical} colors at step $T$ -- colors whose load increased at step $T$ and is higher than the following color load.
Formally, color $\ktau$ is \emph{critical} with respect to vertex $u$ and its $T^{th}$ neighbor if $\ell^{T}(\ktau) > \ell^{T-1}(\ktau)$ and $\ell^{T}(\ktau) > \ell^{T}(\ktau+1)$.
Clearly, in order to upper bound the load at step $T$, it is sufficient to upper bound the load for critical colors $\ktau$ for $T$.
If we let $V_1^\ktau \triangleq \sum_{i=1}^\ktau \ell^T(i)$ be the total load on colors $1,2,\dots,\ktau$ and $V_2^\ktau \triangleq \sum_{i=\ktau+1}^{\delta^T}  \ell^T(i)$ be the total load on colors $\ktau+1,\dots ,\delta^T$, 
we will upper bound the load of color $\ktau$ by 
\begin{equation}
\label{eq:upperboundk}
\ell^T(\ktau) \leq \frac{V^k_1}{\ktau} \leq \frac{\delta^T - V^k_2}{\ktau},
\end{equation}
where the first inequality is due to the monotonicity of the loads, and the second inequality is due to the total load being at most $\delta^T$. 
Therefore, we will upper bound the load by proving a lower 
bound on the index of any critical color, and a lower bound on the total load after this index. 

The next lemma plays a key role in both lower bounds.
We show that for any color $\ktau$ critical at step $T$ and for all steps $A < t \leq T$ during which $\ktau$'s load increases, all colors after $\ktau$ that could be increased (i.e. $\ktau < i \leq \delta^T$) have their load increase by the maximum allowable amount, $\beta/\delta^t$.

\begin{restatable}{lem}{lemtechlem}
	\label{lem:tech_lem}	
	For a color $k$ critical at step $T$, for all $A < t \leq T$ such that  $\ell^{t}(\ktau) > \ell^{t-1}(\ktau)$,
	we have
	$$\ell^{t}(i) - \ell^{t-1}(i) = \beta / \delta^{t} \qquad \forall \ktau < i \leq \delta^t.$$
\end{restatable}

Using the previous lemma, we bound the minimal index of a critical color at step $T$.
\begin{restatable}{lem}{lemtauminindex}
	\label{lem:tau_min_index}
	If $\ktau$ is a critical color at step $T$, then $\ktau > \delta^T \cdot \left(1-1/\beta\right)$.
\end{restatable}

Next, using \Cref{lem:tech_lem} and some useful claims in \Cref{sec:fractional-deferred} we prove a lower
bound on $V^k_2$.
\begin{restatable}{lem}{techlemall}
	\label{lem:tech_lem_all}
	If $\ktau$ is a critical color at step $T$ and	$\ktau^* \geq \max \{ \ktau, \delta^{A}\}$, we have	
	\[
	V_2^k \geq \sum_{j = \ktau+1}^{\delta^{T}} \left( \ell^{T}(j) - \ell^{\ktau^*}(j) \right) \geq \beta \cdot \left(\delta^T - \ktau^* - \ktau \log\,\frac{\delta^T}{\ktau^*}\right).   
	\]
\end{restatable}

\paragraph{Bounding the maximum load.}
Next, we use the previous lemmas in order to bound the maximum load after an assignment of an edge.
Specifically, we will bound the load of $\ell_u$ and $\ell_v$ after coloring the edge $(v,u)$, where $v$ is the newly-arrived vertex. First, it is easy to bound the load of a vertex $v$ for each color after its arrival, 
since we bound each edge-color pair's value $x_{e,c}$ by $\beta/\delta_v^{A_v} \leq \beta/A_v$ at arrival of $v$ (when it has $A_v$ neighbors).
\begin{obs}\label{obs:load_arrival}
	$\ell_{v}^{A_v}(i) \leq \beta$ for all $i\in [\delta^{A_v}]$.
\end{obs}

We next use \Cref{lem:tech_lem_all} and \Cref{eq:upperboundk} to bound the load of previously-arrived vertex $u$.

\begin{restatable}{lem}{boundbigtau}
	\label{lem:bound_bigtau}
	If $\ktau > \delta_u^{A_u}$ is a critical color at step $T$ with respect to $u$, then $\ell_u^T(\ktau) \leq \beta\log\,\frac{\beta}{\beta - 1}$.
\end{restatable}
\begin{restatable}{lem}{boundsmalltau}
	\label{lem:bound_smalltau}	
	If $\ktau \leq \delta_u^{A_u}$ is a critical color at step $T$ with respect to $u$, then  $\ell_u^T(\ktau) \leq \beta^2 - \beta + \beta\log\frac{1}{\beta - 1}.$
\end{restatable}

\paragraph{Upper Bounding \Cref{alg:cap_water_filling}'s Competitive Ratio}\label{sec:upper_bound_bipartite}
We are now ready to bound the competitive ratio of \Cref{alg:cap_water_filling}.
First, we show that  \Cref{alg:cap_water_filling} is $\frac{e}{e-1}$ competitive for one-sided bipartite graphs. 
That is, $G(L, R, E)$ is a bipartite graph and the offline vertices $L$ arrive before the algorithm starts (i.e., $A_u = 0$ for all $u \in L$).
\begin{thm}\label{thm:bipartite-upper-bound}
	For bipartite graphs under one-sided arrivals, \Cref{alg:cap_water_filling} is $\max\{\beta, \beta \log \,\frac{\beta}{\beta - 1}\}$ competitive. Setting $\beta = \frac{e}{e - 1}$, we obtain an $(\frac{e}{e - 1})$-competitive algorithm.
\end{thm}

\begin{proof}
	We bound the load after coloring of edge $(v,u)$, where $v\in R$ is the $T^{th}$ online neighbor of $u$.
	First, we bound the load for any color $i$ of $v$. By \Cref{obs:load_arrival}, we have $\ell_v(i) = \ell^{A_v}_v(i) \leq \beta$.
	For vertex $u$, we have $A_u = \delta^{A_u} = 0$. Thus, by \Cref{lem:bound_bigtau} we have that
	$\max_{i}\ell^T_u(i) \leq \beta\log\,\frac{\beta}{\beta - 1}$.
\end{proof}

Finally, in \Cref{sec:fractional-deferred} we bound our algorithm's competitive ratio on general graphs, proving that it is better than greedy.
\begin{thm}\label{thm: general upper bound}
	For any graph, \Cref{alg:cap_water_filling} is $ \beta^2 - \beta + \beta\log \frac{1}{\beta - 1}$ competitive. Setting $\beta = 1.586$, we obtain a $\generalUB$-competitive algorithm.
\end{thm}
	\section{Online Rounding of Fractional Edge Coloring}\label{sec:online-rounding}

In this section we show how to round fractional edge-coloring algorithms' output online. 
Specifically, we will round fractional edge colorings provided by algorithms which assign at most 
some (small) value $\epsilon$ to each edge-color pair, which we refer to as  \emph{$\epsilon$-bounded} algorithms. (As we shall see, the optimal fractional algorithms we will plug into this rounding scheme both satisfy this property.)  We now state our main technical result of this section: a nearly-lossless rounding process for bounded algorithms on graphs with high enough lower bound on $\Delta$.

\begin{restatable}{thm}{roundingunknown}\label{rounding-unknown}
	For all $\alpha\in [1,2]$ and $\epsilon\leq 1$, if there exists an $\epsilon$-bounded $\alpha$-competitive fractional algorithm $\mathcal{A}$ for bipartite graphs with unknown maximum degree $\Delta\geq \Delta' \geq 2/\epsilon$, then there exists a randomized integral algorithm $\mathcal{A}'$ which is $(\alpha+O(\sqrt[\rootconstant]{(\log n)/\Delta'})$-competitive w.h.p on bipartite graphs of unknown maximum degree $\Delta\geq \Delta' \geq c\cdot \log n$ for some constant $c$.
\end{restatable}

In the end of the section we show how to use this theorem to obtain a $(1+o(1))$-competitive for known $\Delta$.
For now, we note that plugging in our optimal fractional algorithm for unknown $\Delta$ into \Cref{rounding-unknown},\footnote{Strictly speaking, our optimal fractional algorithm, \Cref{alg:cap_water_filling}, is not $2/\Delta'$ bounded. However, setting our initial lower bound on $\Delta$ to be $\Delta'$ in \Cref{line:delta-init} yields a $2/\Delta'$-bounded solution without worsening the competitive ratio. (This is equivalent to adding a dummy star which does not increase the maximum degree.)} we get an optimal randomized algorithm for edge coloring graphs with unknown $\Delta$.

\begin{thm}\label{unknown-delta}
	There exists an $(\frac{e}{e-1} + O(\sqrt[\rootconstant]{(\log n)/\Delta'}))$-competitive algorithm for $n$-vertex bipartite graphs $G$ with unknown maximum degree $\Delta\geq \Delta' \geq c\cdot \log n$ for some absolute constant $c$.
\end{thm}

\noindent\textbf{Remark.} 
The algorithm of \Cref{unknown-delta} requires only a lower bound $\Delta'\leq \Delta$ for some $\Delta' = \omega(\log n)$ in order to output an $(\frac{e}{e-1}+o(1))\cdot \Delta$ coloring, and not the exact value of $\Delta$. 
Alternatively, our algorithm uses $(\frac{e}{e-1}+o(1))\cdot \max\{\Delta,\Delta'\}$ colors for \emph{any} unknown~$\Delta$, where the multiplicative approximation ratio is clearly only worse than $(\frac{e}{e-1}+o(1))$ for small $\Delta < \Delta'$ -- in which case the additive approximation term is only $O(\Delta')$. This result can therefore be read as an asymptotic approximation scheme, trading off between the additive term and the asymptotic competitive ratio.

To describe our rounding scheme, we need the following online rounding scheme of bounded fractional matchings, which motivates our study of bounded fractional edge colorings.

\begin{lem}[Per-Edge Guarantees \cite{cohen2018randomized}]
	\label{peredge-general}
	For all $\epsilon\in [0,1]$, there exists an online dependent rounding algorithm, \marking, which if presented online with a feasible fractional bipartite matching $\vec{x}$ with an (a priori) guarantee $\max_e x_e \leq \epsilon$, outputs a matching $\mathcal{M}$ which matches each edge $e$ with probability
	$$x_e\cdot \left(1-11\sqrt[3]{\epsilon \cdot \log (1/\epsilon)}\right) \leq \Pr[e\in \mathcal{M}] \leq x_e.$$
\end{lem}\color{black}

We now outline our rounding scheme, which consists of \phases, as follows.
For each \phase $i$, let $U_i$ be the uncolored graph at start of \phase $i$. (Initially, $U_1=G$.) We compute an $\alpha$-competitive fractional edge coloring in $U_i$ online. Upon the algorithm's initialization, we sample each of the possible $\alpha\cdot n$ fractional matchings of this fractional coloring, i.i.d with probability $p$. 
We then round and color the sampled fractional matchings in an online fashion, as follows.
Whenever a \emph{sampled} fractional matching becomes non trivial, we assign it a new color. Whenever a new vertex $v$ arrives, for each \phase $i$ in increasing order, we run the next step of \marking for each of the sampled fractional matchings of \phase $i$'s fractional coloring, and color all newly-matched edges with the color assigned to the relevant fractional matching. Finally, we greedily color the remaining uncolored edges of $v$. Setting $p=o(1)$ (guaranteeing few re-colors) and also satisfying $\Delta \cdot p = \omega(\log n)$ (in order to have concentration up to $(1\pm o(1))$ factors on number of colors used), this approach will use roughly $p\cdot \alpha \cdot \Delta(U_i)$ colors for the $i^{th}$ \phase, while decreasing the uncolored subgraph's maximum degree by roughly $p \cdot \Delta(U_i)$, or a $(1-p)$ factor. 
Thus, using $(1/p)\log (1/p)$ \phases yield an uncolored subgraph of maximum degree $p\cdot \Delta$ (using $\alpha\cdot \Delta$ colors), which the greedy algorithm colors using $2p\cdot \Delta$ new colors. This implies \Cref{rounding-unknown}.
    
\subsection{Our Online Rounding Scheme}\label{sec:unknown-delta-algorithm}

Our online rounding scheme, given an $\eps$-bounded fractional edge-coloring algorithm $\mathcal{A}$ which is $\alpha$~competitive on graphs of maximum degree at least $2/\eps$, for $\eps=p^4/(12\log n)$, works as follows. 
Let $p \triangleq \sqrt[\rootconstant]{24(\log n)/\Delta'}$. We use $P \triangleq \numphases$ many \phases.
For \phase $i$, we sample in advance a subset $\mathcal{S}_i$ of all possible color indices, each taken into $\mathcal{S}_i$ with probability $p$. 
Let $U_i$ be the subgraph of edges not colored before \phase $i$. When online vertex $v$ arrives, for each \phase 
$i\in [P]$, we update a fractional coloring $x^{(i)}$ using Algorithm $\mathcal{A}$, based on $v$'s arrival in $U_i$. 
For all sampled $j\in \mathcal{S}_i$ for which $x^{(i)}_j$ (the $j^{th}$ fractional matching of $x^{(i)}$) is non trivial, we use a distinct color $c_{i,j}$ to color edges of a matching $M_{i,j}$ computed online by running \marking on $x^{(i)}_j$. Finally, all remaining uncolored edges of $v$ are greedily colored
 using new colors. This is \Cref{alg:round-unknown}, below.

\begin{algorithm}[t] 
	\caption{Randomized Edge Coloring for Unknown $\Delta$}
	\label{alg:round-unknown}
	
	\begin{algorithmic}[1]
		\Require Online $n$-vertex bipartite graph $G(L,R,E)$ with $\Delta \geq \Delta' \geq c \cdot \log n$, for $c$ a constant TBD. \Statex Parameter $p \triangleq \sqrt[\rootconstant]{(24\log n)/\Delta'}(\leq 1/10)$. An $\epsilon$-bounded fractional online edge-coloring~algorithm~$\mathcal{A}$ which is $\alpha$ competitive on graphs $U$ with $\Delta(U) \geq 2/\epsilon$, for $\epsilon \triangleq (p^4/12\log n)$.
		
		\Ensure Integral $(\alpha+O(p))\cdot \Delta$ edge coloring, w.h.p.
	
		\State for all $i$, set $\mathcal{S}_i\subseteq \lceil \alpha \cdot n\rceil$ to be such that each $j\in \lceil \alpha \cdot n\rceil$ is in $\mathcal{S}_i$ independently with probability~$p$.
		\State for all $i$, denote by $U_i$ the \textbf{online} subgraph of $G$ not colored during \phases  $1,2,\dots,i-1$.

		\For{each arrival of a vertex $v\in R$}
		\For{\phase $i=1,2,\dots,\numphases$} 
		\State $x^{(i)}\leftarrow$ output of Algorithm $\mathcal{A}$ on \textbf{current} $U_i$. \Comment{run next step of $\mathcal{A}$}
		\For{$j \in \mathcal{S}_i$ with $x^{(i)}_{j} \neq \vec{0}$}
		\If{$c_{i,j}$ not set}
		\State set $c_{i,j}$ to be the next unassigned color index.
		\EndIf 
		
		\State $M_{i,j}\leftarrow $ output of \marking run on \textbf{current} $x^{(i)}_{j}$. \Comment{run next step of \marking}
		\If{an edge $e\in M_{i,j}$ is previously uncolored} \Comment{note: $e\ni v$}
		\State color $e$ using color $c_{i,j}$.
		\EndIf
		\EndFor 
		\EndFor
		\State run greedy on uncolored edges of $v$, using colors not assigned during the \phases. \label{line:greedy-unknown}
		\EndFor
	\end{algorithmic}
\end{algorithm}
\vspace{-0.3cm}

\subsection{Analysis}\label{sec:unknonwn-delta-analysis}
We will study changes in the uncolored graph between subsequent \phases and the colors used during the \phases.
For each $i$, let $\Delta_i \triangleq \Delta(U_i)$ be the maximum degree of the online graph not colored by \phase $1,2,\dots,i-1$.
In this section we will show that during each \phase $i$, provided $\Delta_i$ is sufficiently large,  \Cref{alg:round-unknown} uses some $\alpha\cdot \Delta_i\cdot p(1+O(p))$ new colors w.h.p., and obtain an uncolored subgraph $U_{i+1}$ of maximum degree $\Delta_{i+1}= \Delta_i \cdot (1-p \pm O(p^2))$ w.h.p. This will imply a degree decrease at a rate of one per $\alpha+O(p)$ colors used. Repeating this for $\numphases$ \phases, will therefore require $(\alpha+O(p))\Delta$ colors and yield a subgraph of maximum degree $p\cdot \Delta$, which we color greedily with $O(p)\Delta$ new colors, implying \Cref{rounding-unknown}.

To upper bound the number of colors used in \phase $i$, 
we note that the number of non-trivial (i.e., not identically zero) fractional matchings we round in each \iteration is clearly a $p$-fraction of the (at most $\lceil \alpha\cdot \Delta_i \rceil$) non-trivial colors of $x^{(i)}$. Therefore, by standard Chernoff bounds (\Cref{chernoff}), if $\Delta_i$ is large enough, the number of colors in the \phase is small, w.h.p.

\begin{restatable}{lem}{unknownnumcolors}\label{rand-num-colors-used}
	If $\Delta_i \geq (6\log n) / p^3$, then $C_i$, the number of colors used in \phase $i$, satisfies
	$$\Pr\left[C_i\geq \alpha \Delta_i\cdot p \cdot \left(1+p\right)\right] \leq \frac{1}{n^2}.$$
\end{restatable}

\Cref{rand-num-colors-used} upper bounds the number of colors used in \phase $i$ by $\alpha\Delta_i\cdot p\cdot (1+p)$. Our main technical lemma, below, whose full proof is deferred to \Cref{sec:unknown-delta-deferred}, asserts that these colors result in a decrease of roughly $\Delta_i \cdot p$ in the uncolored subgraph's maximum degree during the \phase.

\begin{restatable}{lem}{unknowndegdecrease}\label{rand-max-deg-decrease}
	If $\Delta_i\geq (24\log n)/p^4$, then 
\begin{enumerate}
	\item 
	$\Pr\left[\Delta_{i+1}\leq \Delta_i\cdot (1-p-4p^2)\right] \leq 3/n^3$.
	\item $\Pr\left[\Delta_{i+1}\geq \Delta_i\cdot (1-p+7p^2)\right] \leq 6/n^2$.
\end{enumerate}
\end{restatable}
\begin{proof}[Proof Sketch]
	Let $v$ be a vertex of degree $d_i(v) \geq \Delta_i/2$ in $U_i$. By \Cref{peredge-general} and the $\epsilon$-boundedness of the fractional algorithm $\mathcal{A}$ (and some simple calculations), each edge $e\in U_i$ is matched in $M_{i,j}$ ($j\in \mathcal{S}_i$) with probability $x^{(i)}_{e,j} \cdot (1-O(p))\leq \Pr[e\in M_{i,j}] \leq x^{(i)}_{e,j}$.
	That is, we match $e$ in $M_{i,j}$ with probability close to its sampled ``load'' for this color.
	By Chernoff bounds, as we sample each color of $x^{(i)}$ with probability $p$, the sampled load on $v$'s edges is $d_i(v) \cdot p (1 \pm O(p))$ w.h.p. 
	So, by linearity and another Chernoff bound, the number of times $v$ is matched during the $i^{th}$ \phase satisfies $M_v \leq d_i(v)\cdot p(1+O(p))^2\leq d_i(v)\cdot p(1+O(p))$,
	and $M_v\geq d_i(v)\cdot p(1-O(p))^3 \geq d_i(v)\cdot p (1-O(p))$.
	
	However, $M_v$ also counts repeated matchings of edges of $v$, which do not contribute to $v$'s degree decrease in the uncolored subgraph. We therefore want to bound $R_v$ -- the number of times a previously-colored edge of $v$ is matched during the \phase. 
	By Chernoff's bound and $\epsilon$-boundedness of the fractional algorithm, the load on each edge in the sampled colors $\mathcal{S}_i$, which in expectation is precisely $p$, is $O(p)$ w.h.p.
	So, intuitively, we would expect $R_v = \Theta(p) \cdot M_v$ w.h.p., implying $R_v = \Theta(d_i(v)\cdot p^2)$ w.h.p. Of course, as re-matches are not independent of matches, we cannot simply multiply these expressions this way. However, relying on the theory of negative association (see \Cref{sec:na}), the intuitive claim that $R_v = \Theta(d_i(v)\cdot p^2)$ w.h.p.~can be formalized. We conclude that the degree decrease of vertex $v$ in the uncolored graph during the $i^{th}$ \phase is $M_v - R_v = d_i(v)\cdot p\cdot (1-\Theta(p))$ w.h.p. Taking union bound over all vertices $v$, 
	the lemma follows.
\end{proof}

\Cref{rounding-unknown} now follows from \Cref{rand-num-colors-used} and  \Cref{rand-max-deg-decrease}. We sketch a proof of this theorem and defer its full proof to \Cref{sec:unknown-delta-deferred}.
\begin{proof}[Proof of \Cref{rounding-unknown} (Sketch)]
	Clearly, \Cref{alg:round-unknown} colors all edges of $G$, due to \Cref{line:greedy-unknown}. By definition, all color classes computed are matchings. As we shall show, the number of colors used during the \phases is at most $(\alpha+O(p))\cdot \Delta$ w.h.p., and the greedy algorithm requires some $O(p)\cdot \Delta$ colors w.h.p., implying our claimed result. We outline this proof using a stronger claim than \Cref{rand-max-deg-decrease}.
	
	Suppose instead of \Cref{rand-max-deg-decrease} we had that with high probability $\Delta_{i+1} = \Delta_i\cdot (1-p)$. Then, by induction we would have $\Delta_i = \Delta\cdot (1-p)^i$ and in particular for all $i\leq (1/p)\log (1/p)$ we would have $\Delta_i \geq \Delta \cdot p \geq \Delta'\cdot p$. Taking $p\geq \sqrt[5]{(24\log n)/\Delta'}$ would therefore imply that $\Delta_i \geq \Delta'\cdot p \geq (24\log n)/p^4$, which in turn would allow us to appeal to union bound to prove that $\Delta_i = \Delta\cdot (1-p)^i$ for all $i$, or in other words $\Delta_{i} - \Delta_{i+1} = \Delta_i\cdot p$, and that the number of colors used in each \phase $i$ is at most $C_i\leq \alpha\cdot \Delta_i \cdot p\cdot (1+p)$. Summing over all \phases, this would imply that w.h.p., the number of colors used during the \phases is 
	\begin{align*}
	\sum_i C_i & \leq \sum_i (\alpha + p(1+p)) \cdot (\Delta_{i}-\Delta_{i+1}) \leq (\alpha + p(1+p))\cdot \Delta_0 = (\alpha + p(1+p))\cdot \Delta.
	\end{align*}
	On the other hand, after $(1/p)\log (1/p)$ \phases we would get a final uncolored subgraph of maximum degree  $\Delta\cdot (1-p)^{(1/p) \log (1/p)} \approx \Delta\cdot p$ w.h.p., and so the greedy step of \Cref{line:greedy-unknown} would use at most $2\Delta\cdot p$ colors. Overall,  \Cref{alg:round-unknown} therefore uses at most $(\alpha+O(p))\cdot \Delta$ colors for $p=O(\sqrt[5]{(\log n)/\Delta'})$ and $\Delta\geq 24\log n$. Our more involved bounds are due to the slightly looser bounds for $\Delta_{i+1}$ in terms of $\Delta_i$ in \Cref{rand-max-deg-decrease}. See full proof in \Cref{sec:unknown-delta-deferred} for details.
\end{proof}	

\paragraph{Applications to Known $\Delta$.} 
\Cref{alg:round-unknown} finds applications for \emph{known} $\Delta$, too. In particular, by \Cref{rand-max-deg-decrease} we find that if in each phase $i$ we assign value $1/((1-p+7p^2)^i\cdot \Delta)$ for each edge-color pair, then we obtain a feasible coloring w.h.p., requiring $(1-p+7p^2)^i\cdot \Delta$ colors when the maximum degree is at least $(1-p-4p^2)^i\cdot \Delta$, w.h.p.; i.e., this is a $(1+O(p^2))$-competitive fractional algorithm for uncolored subgraph $U_i$. 
Replacing algorithm $\mathcal{A}$ in \Cref{alg:round-unknown} with this approach then yields, as in the proof of  \Cref{rounding-unknown}, an optimal randomized algorithm for known $\Delta$.
\begin{thm}\label{known-delta-from-rounding}
	There exists a $(1 + O(\sqrt[\rootconstant]{(\log n)/\Delta}))$-competitive algorithm for $n$-vertex bipartite graphs $G$ with known maximum degree $\Delta \geq c\cdot \log n$ for some absolute constant $c$.
\end{thm}

In this section we provided optimal online edge coloring algorithms for known and unknown $\Delta$. In \Cref{sec:known-delta} we improve the $o(1)$ term in the $1+o(1)$ competitive ratio for known $\Delta$. In the following section we present our lower bounds for known and unknown $\Delta$, proving the optimality of our fractional and randomized algorithms, up to $o(1)$ terms.
	\section{Lower Bounds}\label{sec:hardness}
In this section we present our lower bounds for online edge coloring. We start with by noting that for known $\Delta$, the competitive ratio of $(1+o(1))$ we obtain is optimal (up to the exact $o(1)$ term).\footnote{A similar argument implies that $1+o(1)$ competitiveness is impossible on arbitrary multigraphs. See \Cref{sec:multigraphs}.}

\begin{obs}\label{lb-known}
	No randomized online edge coloring algorithm is $(1+o(1/\sqrt{\Delta}))$-competitive.
\end{obs}
\begin{proof}
	By \cite{cohen2018randomized}, no online matching algorithm outputs a matching of expected size $(1-o(1/\sqrt{\Delta}))\cdot n$ in $\Delta$-regular $2n$-vertex bipartite graphs under one-sided arrivals. Given a $(1+\epsilon)$-competitive edge coloring algorithm, we can randomly pick one of the $(1+\epsilon)\cdot \Delta$ color classes upon initialization and output that as our matching. For $\Delta$-regular graphs on $2n$ vertices, which have $\Delta\cdot n$ edges, this results in a matching of expected size $\frac{\Delta\cdot n}{(1+\epsilon)\cdot \Delta} = (1-O(\epsilon))\cdot n$, from which we conclude $\epsilon = \Omega(1/\sqrt{\Delta})$.
\end{proof}

Our main result of the section is a lower bound for unknown $\Delta$ of $\frac{e}{e-1}$ on the competitive ratio of any fractional online algorithm for our relaxation (and by extension, for any randomized algorithm). 
To obtain this result, 
we derive linear constraints that any $\alpha$-competitive fractional online algorithm must satisfy and formulate these constraints as a family of linear programs. 
Specifically, we will rely on the modified fractional edge coloring formulation, where the competitive ratio $\alpha \triangleq \max_{v,c} \sum_{e\ni v} x_{e,c}$ is the maximum \emph{load} of any vertex $v$ for color $c$, and $x_{e,c}\geq 0$ for all $c\in [\Delta]$ and $x_{e,c}=0$ for all $c>\Delta$, for $\Delta$ the \textbf{current} maximum degree. (See \Cref{sec:fractional-def}.)
We then construct feasible dual solutions to these LPs, which by LP duality imply our claimed lower bounds.

\subsection{Matching Lower Bound for Bipartite Graphs}
\label{sec:lower-bipartite}
Our first lower bound concerns fractionally edge coloring bipartite graphs.
\begin{thm}\label{thm: lower-bipartite}
	No fractional online edge coloring algorithm is better than $\frac{e}{e-1}$ competitive on bipartite graphs under one-sided arrivals.
\end{thm}
\begin{proof}
	Consider the following construction. For any $m$, we construct a bipartite graph $G_m = (L_m, R_m, E_{m})$, where $L_m$ is the offline side and $R_m$ is the online side. The offline side, $L_m$, contains $m!$ vertices, denoted by $v_1, \cdots, v_{m!}$. The online side, $R_m$, arrives over $m$ phases. In phase $k$ ($k \in [m]$), some $m!/k$ vertices of degree $k$ arrive. Each vertex $u_{i}$ which arrives in phase $k$ ($i\in [m!/k]$) neighbors offline vertices $v_{i}, v_{m!/k + i}, \cdots, v_{m!(k-1)/k + i}$. We can see that each offline vertex has exactly one more neighbor in phase $k$ and the maximum degree in phase $k$ is exactly $k$. See \Cref{fig:hard-instance-bipartite} for an illustrative example. The algorithm will have to be $\alpha$ competitive after each phase, as the adversarial sequence can ``terminate early'', after essentially presenting disjoint copies of $G_{m'}$ for some $m'\leq m$.
	
	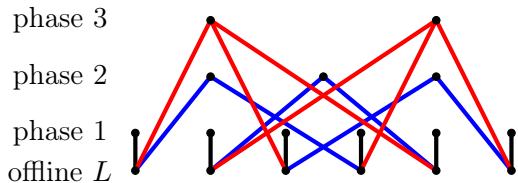
\begin{figure}[!htbp]
		\centering
		\begin{tikzpicture}[scale=0.5]

		\draw[ultra thick] (0, 0) -- (0,1);
		\draw[ultra thick]  (2, 0) -- (2, 1);
		\draw[ultra thick]  (4, 0) -- (4,1);
		\draw[ultra thick]  (6, 0) -- (6, 1);
		\draw[ultra thick]  (8, 0) -- (8,1);
		\draw[ultra thick]  (10, 0) -- (10, 1);

		\draw[ultra thick, blue]  (2, 2.5) -- (0, 0);
		\draw[ultra thick, blue]   (2, 2.5) -- (6, 0);
		\draw[ultra thick, blue]  (5, 2.5) -- (2, 0);
		\draw[ultra thick, blue]   (5, 2.5) -- (8, 0);
		\draw[ultra thick, blue]  (8, 2.5) -- (4, 0);
		\draw[ultra thick, blue]  (8, 2.5) -- (10, 0);

		\draw[ultra thick, red]  (2, 4) -- (0, 0);
		\draw[ultra thick, red]  (2, 4)  -- (4, 0);
		\draw[ultra thick, red]  (2, 4)  -- (8, 0);
		\draw[ultra thick, red]  (8,4) -- (2, 0);
		\draw[ultra thick, red]  (8,4) -- (6, 0);
		\draw[ultra thick, red]  (8,4) -- (10, 0);

		\draw [fill = black] (0,0) circle [radius=0.1];	
		\draw [fill = black] (2,0) circle [radius=0.1];	
		\draw [fill = black] (4,0) circle [radius=0.1];
		\draw [fill = black] (6,0) circle [radius=0.1];	
		\draw [fill = black] (8,0) circle [radius=0.1];	
		\draw [fill = black] (10,0) circle [radius=0.1];
		
		\draw [fill = black] (2,1) circle [radius=0.1];	
		\draw [fill = black] (4,1) circle [radius=0.1];	
		\draw [fill = black] (6,1) circle [radius=0.1];
		\draw [fill = black] (8,1) circle [radius=0.1];	
		\draw [fill = black] (10,1) circle [radius=0.1];	
		\draw [fill = black] (0,1) circle [radius=0.1];	
		
		\draw [fill = black] (2,2.5) circle [radius=0.1];	
		\draw [fill = black] (5,2.5) circle [radius=0.1];	
		\draw [fill = black] (8,2.5) circle [radius=0.1];
		
		\draw [fill = black] (2,4) circle [radius=0.1];	
		\draw [fill = black] (8,4) circle [radius=0.1];	
		
		\node at (-2, 0){offline $L$};
		\node at (-2, 1){phase 1};
		\node at (-2, 2.5){phase 2};
		\node at (-2, 4){phase 3};
		
		\end{tikzpicture}
		\caption{The hard instance for bipartite graphs for $m = 3$.}
		\label{fig:hard-instance-bipartite}		
	\end{figure}

	We use $x_{kj}\triangleq \frac{\sum_{e\in \textrm{phase } k}x_{e,j}}{|\{e\in \textrm{phase } k\}|}$ to denote the average assignment of color $j$ to edges of phase $k$.

	The average load for online vertices of phase $k$ for color $j$ is $k\cdot x_{kj}$, as each such online vertex has $k$ edges. Consequently, as their average load is at most $\alpha$, we have the following constraints.
	\begin{align}
	k \cdot x_{kj} \leq \alpha \quad \forall 1 \leq j \leq k. \label{eq:bipartite-1}
	\end{align}
	Moreover, since each offline vertex has one more edge during phase $k$, the average assignment to all edges should cover all edges of phase $k$, implying the following constraint.
	\begin{align}
	\sum_{j = 1}^{k}x_{kj} \geq 1 \quad \forall k. \label{eq:bipartite-2}
	\end{align}
	Finally, as the load of all offline vertices (which have only one edge in phase $k$) for any color $j$ cannot exceed $\alpha$ (and so neither can their average), we have the following constraint.
	\begin{align}
	\sum_{k = j}^{m} x_{kj} \leq \alpha \quad \forall j. \label{eq:bipartite-3}
	\end{align}
	
	Combining constraints \eqref{eq:bipartite-1}-\eqref{eq:bipartite-3}, yields the following linear program $\text{LP}_{m}$, which lower bounds $\alpha$.
	\begin{align*}
	\textrm{LP}_{m}&\triangleq\min \alpha\\
	\sum_{j = 1}^{k}x_{kj} & \geq 1  \qquad\qquad 1 \leq k\leq m\\
	k\cdot x_{kj} &\leq \alpha \qquad\qquad  1 \leq j \leq k \leq m\\
	\sum_{k = j}^{m}x_{kj} &\leq \alpha \qquad\qquad  1 \leq j \leq m \\
	x_{kj} &\geq 0 \qquad\qquad  1 \leq j \leq k \leq m.
	\end{align*}

	We construct a series of dual feasible solutions to lower bound $\alpha$. First, the dual LP is as follows.
	\begin{align*}
	\max \sum_{k = 1}^{m}y_{k}& \\	
	\sum_{k = 1}^{m}\sum_{j =1 }^{k}z_{kj} + \sum_{j = 1}^{m}w_{j} &\leq 1 \\
	- k\cdot z_{kj} - w_{j} + y_{k} &\leq 0 \qquad\qquad  1\leq j\leq k \leq m\\
	y_k, w_j, z_{kj} &\geq 0 \qquad\qquad 1\leq j\leq k \leq m.
	\end{align*}
	Let $c(m)\triangleq \lfloor m/e \rfloor$. We know that $\lim_{m \rightarrow \infty} c(m)/m \rightarrow 1/e$. Let $t \triangleq 1/(m + 1 + c(m)\cdot (H_{c(m)} - H_m))$, where $H_k \triangleq \sum_{i=1}^k 1/k$ satisfies $\lim_{m\rightarrow \infty} H_{c(m)} - H_m \rightarrow \log(c(m)/m) \rightarrow -1$. We construct a feasible dual solution as follows: We let $y_1 = \cdots y_m = t$, and
	\begin{gather*}
		w_j = \left\{
		\begin{matrix}
		t & 1 \leq j \leq c(m)\\
		0 & \text{otherwise}
		\end{matrix}
		\right.
	\end{gather*}
	\vspace{-0.3cm}
	\begin{gather*}
	z_{kj} = \left\{
	\begin{matrix}
	t/k & c(m) + 1 \leq j \leq k \leq m\\
	0 & \text{otherwise.}
	\end{matrix}
	\right.
	\end{gather*}
	For any $1 \leq j\leq k\leq m$, we have that $k\cdot z_{kj} + w_{j} = t = y_k$. For the first dual constraint, we have
	\begin{align*}
	\sum_{k = 1}^{m}w_{k} + \sum_{k = 1}^{m}\sum_{j =1 }^{k}z_{kj} &= c(m)\cdot t + \sum_{k = c(m)}^{m}\left(\frac{k - c(m)}{k}\right)\cdot t \\
	& = c(m)\cdot t + (m - c(m) + 1)\cdot t - c(m)\cdot t\cdot (H_m - H_{c(m)}))\\
	& = \left(m+1 + c(m)\cdot (H_{c(m)} - H_m\right)\cdot t = 1.
	\end{align*}

	The above is therefore a feasible dual solution, of value 
	\begin{gather*}
	\sum_{i=k}^{m}y_k = m\cdot t = \frac{m}{m + c(m)\cdot (H_{c(m)} - H_m)} = \frac{1}{1 + \frac{c(m)}{m}\cdot (H_{c(m)} - H_m)}.
	\end{gather*}	
	When $m\rightarrow \infty$, this tends to $\frac{1}{1 - 1/e}=\frac{e}{e-1}$.
	Consequently, $\lim_{m \rightarrow \infty} \text{LP}_m \geq e/(e-1)$, implying our claimed lower bound for fractional online edge coloring of bipartite graphs.
\end{proof}
\paragraph{Making the Graph Dense}
\label{sec:dense-graph}
The above construction yields a sparse graph, as the number of vertices in this graph, $n = m! + m!(1 + \frac{1}{2} + \cdots + \frac{1}{m}) \approx m!\log m$, is exponential in its maximum degree, $m$. However, the following change yields a dense graph where the same lower bound still holds. Fix any integer $t > 0$, in the hard instance, we replace each vertex with $t$ identical copies, and correspondingly, connecting all copies of pairs $(u,v)$  which are adjacent in the sparse graph.
The obtained graph is still bipartite and the maximum degree and the number of vertices both increase by a factor of $t$, to $t\cdot m$ and $t\cdot m!\log m$, respectively. Since we can take $t$ to be arbitrarily large, the graph has maximum degree as high as $\Omega(n)$. In order to show that the lower bound still holds, we only need to slightly change the meaning of $x_{kj}$ to be the average assignment of colors $(j-1)t + 1, (j-1)t+2,\dots jt$ during phase $k$. Constraints  \eqref{eq:bipartite-1}-\eqref{eq:bipartite-3} still hold with this new meaning in the denser graph. Thus, we conclude that \Cref{thm: lower-bipartite} holds for graphs of arbitrarily high degree.

\vspace{-0.1cm}
\subsection{Lower Bound for General Graphs}
\label{sec:lower-general}
Next, we present a lower bound for general graphs. The lower bound is based on the construction for bipartite graphs, but with more alterations. More specifically, recall that in the construction for bipartite graphs, when the online vertices of phase $k$ arrive, we always connect them to $k$ offline vertices. However, in general graphs, we have more freedom. In phase $k$, there can be two possible futures: in one we continue the sequence for bipartite graphs; in the other we connect all vertices which arrive during phases $k,k+1,\dots$ to the vertices which arrived in phase $k - 1$. This example yields a lower bound of $\generalLB$, showing a separation between bipartite and general graphs.\footnote{In \Cref{sec:tight-example}, we show this example is a tight instance for \Cref{alg:cap_water_filling}, which is $\generalUB$ competitive on it.} We state this lower bound here and defer its proof to \Cref{sec:hardness-omitted}. 
\begin{restatable}[]{thm}{lowergeneral}
	\label{thm:lower-general}
	No fractional online edge coloring algorithm is better than $\generalLB$ competitive in general graphs.
\end{restatable}

\paragraph{Acknowledgements} The authors would like to thank Marek Elias, Seffi Naor and Ola Svensson for comments on an earlier draft of this paper, which helped improve its presentation. The authors would also like to thank Lex Schrijver for asking about extension of our results to multigraphs.
	\section*{Appendix}
	\appendix
	\section{Bad Examples for Natural Algorithms}
\label{sec:bad-example}
In this section we present bad examples for a variety of natural edge coloring algorithms for known and unknown $\Delta$.

\subsection{Repeated Maximal Online Matching}
\label{sec:bad-example-greedy}
Here we give a bad example which shows that a family of natural online algorithms, i.e., algorithms that iteratively find a maximal matching, are no better than $2$-competitive, even on dense bipartite graphs under one-sided arrivals. Notice that this family of algorithms includes the natural extension of the \textsc{ranking} algorithm. I.e., iteratively find the maximal matching via the optimal, $(1-1/e)$-competitive, online matching algorithm, \textsc{ranking} -- an approach which at first glance one might guess yields an $\approx(e/(e-1))$-competitive edge coloring.

The bad example is as follows. The graph is made up of $\Delta$ stars with $\Delta-1$ leaves each, with the stars' centers, which are offline vertices, connected to a common vertex $v$, which is the last offline vertex to arrive. It is easy to see that that any algorithm that repeatedly uses a maximal online matching for each color $c=1,2,\dots$ would color each star's edges with colors $1,2, \cdots, \Delta-1$ following the star's $\Delta-1$ leaves' arrivals. The $\Delta$ edges of $v$ therefore require a further $\Delta$ colors. Consequently, such an algorithm would use $2\Delta-1$ colors and is thus $2$ competitive. Adding $n-\Delta^2-1$ isolated dummy nodes with no edges, we get an example with $n$ nodes and maximum degree $\Delta=O(\sqrt n)$. This example therefore rules out this natural family of online edge coloring algorithms for all $\Delta=O(\sqrt n)$.

\subsection{Repeatedly Running Marking}\label{sec:bad-example-marking}
Our online rounding scheme for fractional edge colorings of \Cref{sec:online-rounding} applies \marking to multiple fractional edge colorings. 
For known $\Delta$, this is done by running \marking on some fractional matchings assigning values $\frac{1}{\Delta'}$, which we refer to as $\marking_{\Delta'}$, for increasingly smaller value of $\Delta'$.
Here we show an example underlying the need for rounding multiple edge colorings. In particular, we show that simply iteratively coloring the matching output by $\marking_\Delta$ on the uncolored subgraph -- i.e., rounding one trivial edge coloring -- results in suboptimally many colors.

To be precise, for $c=1,2,\dots, \Delta$, the algorithm considered computes $M_c$, the $c$-th color class, by running $\marking_\Delta$ in the subgraph not colored by the first $c-1$ colors, $U_c \triangleq G\setminus \bigcup_{c'=1}^{c-1} M_{c'}$ (and then reverts to some other algorithm on the uncolored graph, say greedy).
For simplicity (though this is too good be true), let us assume that $\marking_\Delta$ matches each edge $e$ with probability \emph{precisely} $1/\Delta$ in $M_c$ if $e\in U_{c}$. 
Consider a star graph of degree $\Delta$, with the star's center arriving last. Denote by $p_{e,c} \triangleq \Pr[e\in M_c]$ the probability that $e$ is colored $c$. Then, if we run $\marking_\Delta$ for $\Delta$ phases on the uncolored graph, the probabilities $p_{e,c}$ satisfy the recurrence relation
$p_{e,c} = \frac{1}{\Delta}\cdot \left(1-\sum_{c'<c} p_{e,c'}\right).$
But this recurrence captures non-empty bins in a balls and bins process with $\Delta$ balls (colors) thrown into $\Delta$ bins (edges). Specifically, $p_{e,c}$ is the probability of $c$ being the first ball occupying bin $e$. The expected number of unoccupied bins in the above process, $\Delta - \sum_e \sum_c p_{e,c}$, is $(\frac{1}{e}+o(1))\Delta$, so this process results in at least $\frac{\Delta}{e} - o(\Delta)$ uncolored edges in the star after $\Delta$ colors used. Consequently, this approach would use at least $(1+\frac{1}{e} - o(1)) \Delta$ colors, even on a star of maximum degree $\Delta$.

The above bad example rules out running $\marking_\Delta$ for the first $\Delta$ colors and then resorting to some other algorithm (say, greedy). More generally, extending this idea to any prefix of colors computed using $\marking_\Delta$ and then reverting to greedy can be similarly shown to be suboptimal. For example, standard coupon collector arguments show that the extreme approach of repeatedly running $\marking_\Delta$ in $U_c$ for $c=1,2,\dots$ until all edges are colored (i.e., without running greedy) requires at least $\D \log \Delta$ colors(!), even on a star of maximum degree $\Delta$ whose center arrives last.

\subsection{Bad Examples for (Unbounded) Water Filling}
\label{sec:bad-example-wf}
In this section, we will give bad examples to rule out a natural candidate algorithm for fractional edge coloring; i.e., \textsc{water filling}. It is easy to see that we can make the level of a color arbitrary large if we only do \textsc{water filling} on one side. On the other hand, the following algorithm is a natural extension of this algorithm which only conducts \textsc{water filling} on the \emph{maximum} of the two endpoints' loads for any edge $(u,v)$ which arrives. More formally, the algorithm is as follows.

\begin{algorithm}[H] 
	\caption{\textsc{water filling}}
	\label{water-filling}
	\begin{algorithmic}[1]
		\Require Online graph $G(V,E)$ with unknown $\Delta(G)$ under vertex arrivals.\;
		
		\Ensure Fractional edge coloring $\{x_{e,c} \mid e\in E$, $c\in [\Delta(G)]\}$.

		\For{each arrival of a vertex $v$}
		\State $\Delta \leftarrow \max\{d(u)\mid u\in V\}$. \Comment{$\Delta$ = \textbf{current} max.~degree} \label{line:delta-init-wf}
		\For{each $e=(u,v)\in E$}
		\While{$\sum_{c\in [\Delta]}x_{e,c} < 1$} \Comment{initially, $x_{e,c}=0$}
		
		\State let $\mathit{C} \triangleq \{c \in \mathit{U} \mid \min_{c\in \D} \max\{L_u(c), L_v(c)\}\}$. \Comment{``currently active'' colors for $e$}
		\For{all $c \in \mathit{C}$}
		\State increase $x_{e,c}$ continuously. \Comment{update $L_u(c),L_v(c)$ and $\mathit{A}$.}
		\EndFor								
		\EndWhile
		\EndFor
		\EndFor
	\end{algorithmic}
\end{algorithm}

Our first observation here is that \textsc{water filling} is $2$ competitive, even under edge arrivals.
\begin{claim}
	The \textsc{water filling} algorithm is at most $2$ competitive under adversarial edge arrivals. 
\end{claim}
\begin{proof}
	We only need to prove $\ell \leq 2$ following updates due to the arrival of an edge $(u,v)$. This inequality holds because $$2\Delta \geq \sum_{i = 1}^{\D}\ell_{u}(i) + \sum_{i = 1}^{\D}\ell_{v}(i) = \sum_{i = 1}^{\D}(\ell_{u}(i) + \ell_{v}(i)) \geq \sum_{i = 1}^\D \max\{\ell_{u}(i), \ell_{v}(i)\} \geq \Delta \cdot \ell,$$ where the last inequality holds since we run \textsc{water filling} on the maximum of $u$ and $v$'s loads.
\end{proof}
By the above, \textsc{water filling} is always at most $2$ competitive. We will show that there exists a hard instance on which \textsc{water filling} is $2$ competitive, even under one-sided vertex arrivals in bipartite graphs. We start first with a bad example for vertex arrivals in general graphs, which illustrates the weakness of \textsc{water filling}, and also motivates the hard instance given in \Cref{sec:hardness}. 

\paragraph{A Bad Example for General Graphs.}
Consider a tree of height $n + 1$. The root locates in the first level. Each vertex in level $k$ has $n - k + 1$ children. In the online process, vertices arrive from level $n + 1$ down to $1$. See \Cref{fig:lower-wtf-general} for an example.
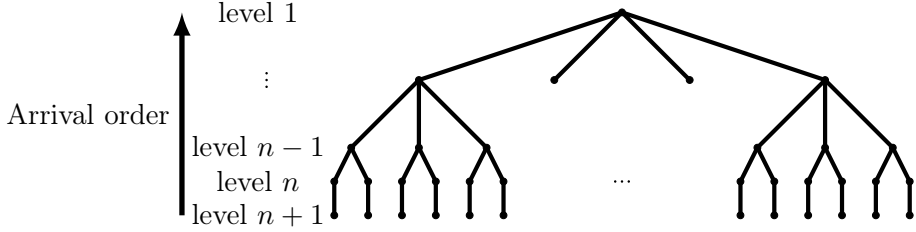
\begin{figure}[!htbp]
	\begin{tikzpicture}[scale = 0.45][!htbp]
	\draw[ultra thick] (0, 0) -- (0,1);
	\draw[ultra thick]  (1, 0) -- (1, 1);
	\draw[ultra thick]  (2, 0) -- (2,1);
	\draw[ultra thick]  (3, 0) -- (3, 1);
	\draw[ultra thick]  (4, 0) -- (4,1);
	\draw[ultra thick]  (5, 0) -- (5, 1);
	\draw[ultra thick]  (0.5, 2) -- (0, 1);
	\draw[ultra thick]  (0.5, 2) -- (1, 1);
	\draw[ultra thick]  (2.5, 2) -- (2, 1);
	\draw[ultra thick]  (2.5, 2) -- (3, 1);
	\draw[ultra thick]  (4.5, 2) -- (4, 1);
	\draw[ultra thick]  (4.5, 2) -- (5, 1);
	\draw[ultra thick]  (2.5, 4) -- (0.5, 2);
	\draw[ultra thick]  (2.5, 4) -- (2.5, 2);
	\draw[ultra thick]  (2.5, 4) -- (4.5, 2);
	
	\draw [fill = black] (0,0) circle [radius=0.1];	
	\draw [fill = black] (1,0) circle [radius=0.1];	
	\draw [fill = black] (2,0) circle [radius=0.1];
	\draw [fill = black] (3,0) circle [radius=0.1];	
	\draw [fill = black] (4,0) circle [radius=0.1];	
	\draw [fill = black] (5,0) circle [radius=0.1];	
	\draw [fill = black] (0,1) circle [radius=0.1];	
	\draw [fill = black] (1,1) circle [radius=0.1];	
	\draw [fill = black] (2,1) circle [radius=0.1];
	\draw [fill = black] (3,1) circle [radius=0.1];	
	\draw [fill = black] (4,1) circle [radius=0.1];	
	\draw [fill = black] (5,1) circle [radius=0.1];	
	\draw [fill = black] (0.5,2) circle [radius=0.1];	
	\draw [fill = black] (2.5,2) circle [radius=0.1];	
	\draw [fill = black] (4.5,2) circle [radius=0.1];	
	\draw [fill = black] (2.5,4) circle [radius=0.1];

	\draw[ultra thick] (12, 0) -- (12,1);
	\draw[ultra thick]  (13, 0) -- (13, 1);
	\draw[ultra thick]  (14, 0) -- (14,1);
	\draw[ultra thick]  (15, 0) -- (15, 1);
	\draw[ultra thick]  (16, 0) -- (16,1);
	\draw[ultra thick]  (17, 0) -- (17, 1);
	\draw[ultra thick]  (12.5, 2) -- (12, 1);
	\draw[ultra thick]  (12.5, 2) -- (13, 1);
	\draw[ultra thick]  (14.5, 2) -- (14, 1);
	\draw[ultra thick]  (14.5, 2) -- (15, 1);
	\draw[ultra thick]  (16.5, 2) -- (16, 1);
	\draw[ultra thick]  (16.5, 2) -- (17, 1);
	\draw[ultra thick]  (14.5, 4) -- (12.5, 2);
	\draw[ultra thick]  (14.5, 4) -- (14.5, 2);
	\draw[ultra thick]  (14.5, 4) -- (16.5, 2);
	
	\draw [fill = black] (12,0) circle [radius=0.1];	
	\draw [fill = black] (13,0) circle [radius=0.1];	
	\draw [fill = black] (14,0) circle [radius=0.1];
	\draw [fill = black] (15,0) circle [radius=0.1];	
	\draw [fill = black] (16,0) circle [radius=0.1];	
	\draw [fill = black] (17,0) circle [radius=0.1];	
	\draw [fill = black] (12,1) circle [radius=0.1];	
	\draw [fill = black] (13,1) circle [radius=0.1];	
	\draw [fill = black] (14,1) circle [radius=0.1];
	\draw [fill = black] (15,1) circle [radius=0.1];	
	\draw [fill = black] (16,1) circle [radius=0.1];	
	\draw [fill = black] (17,1) circle [radius=0.1];	
	\draw [fill = black] (12.5,2) circle [radius=0.1];	
	\draw [fill = black] (14.5,2) circle [radius=0.1];	
	\draw [fill = black] (16.5,2) circle [radius=0.1];	
	\draw [fill = black] (14.5,4) circle [radius=0.1];	
	
	\draw[ultra thick]  (8.5, 6) -- (2.5, 4);
	\draw[ultra thick]  (8.5, 6) -- (14.5, 4);
	\draw[ultra thick]  (8.5, 6) -- (6.5, 4);
	\draw[ultra thick]  (8.5, 6) -- (10.5, 4);
	\draw [fill = black] (8.5,6) circle [radius=0.1];
	\draw [fill = black] (6.5,4) circle [radius=0.1];
	\draw [fill = black] (10.5,4) circle [radius=0.1];

	\draw [fill = black] (8.3,1) circle [radius=0.02];
	\draw [fill = black] (8.5,1) circle [radius=0.02];
	\draw [fill = black] (8.7,1) circle [radius=0.02];
	\node at (-2.25, 0){level $n+1$};
	\node at (-2.25, 1){level $n$};
	\node at (-2.25, 2){level $n-1$};	
	\draw [fill = black] (-2,3.8) circle [radius=0.02];
	\draw [fill = black] (-2,4.0) circle [radius=0.02];
	\draw [fill = black] (-2,4.2) circle [radius=0.02];
	\node at(-2.25, 6){level $1$};
	\draw[line width=2pt, -latex] (-4.5,0) -- node[auto] {Arrival order} (-4.5,6);
	
	\end{tikzpicture}
	\caption{The hard instance for (unbounded) water filling.}
	\label{fig:lower-wtf-general}
\end{figure}

The following claim can be proven by induction and we omit the proof here.
\begin{claim}\label{claim:wt-general}
	The loads for all vertices in level $k$ when they first arrive are
	\begin{itemize}
		\item If $(n - k + 1)$ is odd, the load is $\left(\frac{2(n - k + 1)}{(n - k + 2)}, \cdots, \frac{2(n - k + 1)}{(n - k + 2)}, 0, \cdots, 0 \right)$.
		\item If $(n - k + 1)$ is even, the load is $\left(\frac{2n - 2k + 1}{(n - k + 1)}, \cdots, \frac{2n - 2k + 1}{(n - k + 1)}, \frac{1}{(n - k + 1)}, \cdots, \frac{1}{(n - k + 1)} \right)$.
	\end{itemize}
	When $n$ is large and we set $k = 1$, then the competitive ratio goes to $2$.
\end{claim}

\paragraph{A Bad Example for Bipartite Graphs.}
Notice that the bad example above is a bipartite graph, but vertices can arrive from either side. Here we construct a slightly more complicated bipartite example under one-sided arrivals on which \textsc{water filling} is $2$ competitive.  In the following claim, we use $(a, b)$ to denote $(
\underbrace{a, \cdots, a}_{\D/2}, \underbrace{b, \cdots, b}_{\D/2})$. 

\begin{claim}
	For large enough $\D$, there exists a sequence  of online vertices such that $\boldmath{(2 - 4/\D, 4/\D)}$ is achievable for offline vertices.
\end{claim}
\begin{proof}
	We only sketch the sequence here.
	\[
	\boldmath{(1, 0) \stackrel{(a)}{\rightarrow} (2/\D, 1 + 2/\D) \stackrel{(b)}{\rightarrow} (1 + 2/\D, 0) \rightarrow \stackrel{(c)}{\cdots} \rightarrow (2 - 4/\D, 0)}
	\]
The initial state (1, 0) is achieved by connecting an offline vertex $v$ to $\D/2$ online vertices one by one.~(a) is reached by having an online vertex $u$ neighbor $\D/2 + 1$ offline vertices $(1, 0)$ and one offline vertice $(1 + 1/\D, 1/\D)$, which can be produced in the former step. While~(b) is achievable by connecting $(2/\D, 1 + 2/\D)$ (online) and $(1,0)$ (offline). Finally, we repeat the above process in (a)(b) to get $(2 - 4/\D, 0)$. When $\D$ is large enough, we conclude that the water filling algorithm is exactly $2$ competitive.
\end{proof}
	\section{Omitted Proofs of \Cref{sec:fractional-algorithm}}\label{sec:fractional-deferred}

Here we provide the missing proofs of lemmas whose proof was deferred from \Cref{sec:fractional-algorithm}, restated here for ease of reference.

\lemtechlem*

\begin{proof}
	Suppose there exists $A < t \leq T$ such that  $\ktau+1 \leq \delta^t$ and $\ell^{t}(\ktau+1) - \ell^{t-1}(\ktau+1) < \beta / \delta^{t}$ and $\ell^{t}(\ktau) - \ell^{t-1}(\ktau) > 0$, then we can immediately derive that $\ell^{t}(\ktau) = \ell^{t}(\ktau+1)$, since $\ktau$ and $\ktau+1$ are active at the end of the iteration. But by \Cref{obs:monotonicity} we know that $\ell^{T}(\ktau) = \ell^{T}(\ktau+1)$ -- a contradiction. Finally, 
	$\ell^{t}(i) - \ell^{t-1}(i) \geq \ell^{t}(\ktau+1) - \ell^{t-1}(\ktau+1)$ for all 
	$\ktau < i \leq \delta^t$ by \Cref{obs:monotonicity}. 	
\end{proof}

\lemtauminindex*
\begin{proof}
	By \Cref{lem:tech_lem}, $\ell^{T}(\ktau) > \ell^{T}(\ktau+1)$ and 
	$\ell^{T}(\ktau) > \ell^{T-1}(\ktau)$ imply 
	$\ell^{T}(i) - \ell^{T-1}(i) = \beta / \delta^{T}$, for $\ktau + 1 \leq i \leq \delta^{T}$. Hence, if $\ktau \leq \delta^{T}\cdot \left(1-\frac{1}{\beta}\right)$, we would obtain
	$$\sum_{i = 1}^{\ktau} (\ell^{T}(i) - \ell^{T - 1}(i)) = 1 - \sum_{i = \ktau + 1}^{\delta^{T}} (\ell^{T}(i) - \ell^{T - 1}(i)) = 1 - (\delta^{T} - \ktau) \beta/\delta^{T} < 1 - (\beta/\delta^{T}) \cdot (\delta^{T}/\beta) = 0,$$  
	which would imply $\ell^{T}(\ktau) = \ell^{T-1}(\ktau)$ -- contradicting the fact that $k$ is critical.
\end{proof}

In order to lower bound $V^2_k$, we first prove the following two useful claims.

\begin{claim}
	\label{cor:tech_lem}
	If $\ktau$ is a critical color at step $T$, then for any $j > \ktau$ and for any $\T \geq A$
	\[\ell^{T}(j) - \ell^{\T}(j) = \sum_{\substack{\T < t \leq T \\ \delta^t \geq j}} \frac{\beta}{\delta^t}.
	\]
\end{claim}

\begin{proof}
	We prove that for any $t\geq A$ and $\delta^{t} \geq \ktau$, then $\ell^{t}(\ktau) > \ell^{t-1}(\ktau)$.
	Assume not, then we have $$1 = \sum_{i = 1}^{\delta^{t}}(\ell^{t}(i) - \ell^{t-1}(i)) = \sum_{i = \ktau + 1}^{\delta^{t}}(\ell^{t}(i) - \ell^{t}(i-1)) \leq (\delta^{t} - \ktau)\cdot \beta/\delta^{t} \leq (\delta^{T} - \ktau)\cdot \beta/\delta^{T} < 1.$$
	
	Where that last inequality is since  $\ktau > (1 - 1/\beta)\delta^{T}$, by \Cref{lem:tau_min_index}.
	Therefore, by \Cref{lem:tech_lem}, we have $\ell^{t}(j) - \ell^{t - 1}(j) = \beta/\delta^t$ for $j \leq \delta^t$.
	Consequently, 
	\begin{align*}
	\ell^{T}(j) - \ell^{\T }(j)  & = \sum_{t = \T  + 1}^{T} (\ell^{t}(j) - \ell^{t - 1}(j)) = \sum_{t = \T  + 1}^{T}\mathbb I\{\delta^{t} \geq j\} (\ell^{t}(j) - \ell^{t - 1}(j)) = \sum_{\substack{\T < t \leq T \\ \delta^{t} \geq j}} \frac{\beta}{\delta^{t}}.\qedhere 
	\end{align*}
\end{proof}

Next, we bound the \emph{total} load on the colors after a critical color $\ktau$.

\begin{claim}
	\label{cla:tech_lem_all}
	If $\ktau$ is a critical color at step $T$, then for any  $\T \geq A$
	\[
	\sum_{j = \ktau+1}^{\delta^{T}} \left(\ell^{T}(j) - \ell^{\T}(j)\right) \geq \sum_{j = \T + 1}^{\delta^T} \beta\cdot\frac{\delta^j - \ktau}{\delta^j}.
	\]
\end{claim}

\begin{proof}
	By \Cref{cor:tech_lem}, we have
	\begin{align*}
	\sum_{i=\ktau+1}^{\delta^T} \left(\ell^{T}(i) - \ell^{\T}(i)\right) 
	&\geq \sum_{i=\ktau+1}^{\delta^T} \sum_{\substack{\T + 1 \leq j \leq \delta^T \\ \delta^{j} \geq i}} \frac{\beta}{\delta^j}\\
	&= \sum_{j = \T + 1}^{\delta^T} \sum_{\delta^j \geq i \geq \ktau} \frac{\beta}{\delta^j}\\
	&= \sum_{j = \ktau^* + 1}^{\delta^T} \beta\cdot \frac{\delta^j - \ktau}{\delta^j}.\qedhere
	\end{align*}
\end{proof}

We now ready to prove the main lower bound volume lemma.
\techlemall*
	\begin{proof}		
	Substituting $\T$ with $\ktau^*$ in \Cref{cla:tech_lem_all} (note that,
	$\ktau^* \geq \delta^A \geq A$), we have
	\begin{align*}
	\sum_{j = \ktau+1}^{\delta^{T}} \left(\ell^{T}(j) - \ell^{\ktau^*}(j)\right)
	&= \sum_{j = \ktau^* + 1}^{\delta^T} \beta\cdot \frac{\delta^j - \ktau}{\delta^j}\\
	&\geq \sum_{j = \ktau^* + 1}^{\delta^T} \beta\cdot\frac{j - \ktau}{j}\\
	&\geq \beta\cdot (\delta^T - \ktau^*) - \beta \cdot \ktau\log\,\frac{\delta^T}{\ktau^*}\\	
	&= \beta \cdot \left(\delta^T - \ktau^* - \ktau\log\,\frac{\delta^T}{\ktau^*}\right),
	\end{align*}
	where the first inequality is since  $\delta^j \geq j$.
\end{proof}

\boundbigtau*
\begin{proof}
	As $k$ is critical at step $T$, by  \Cref{lem:tech_lem_all}, taking $\ktau^* = \ktau > \delta^A$, we have 
	$$V_2^k = \sum_{i=\ktau+1}^{\delta^T} \ell^{T}(i) \geq \sum_{i=\ktau+1}^{\delta^T} \left(\ell^{T}(i) - \ell^{\ktau}(i) \right)   \geq \beta\cdot \left(\delta^T - \ktau - \ktau\log\,\frac{\delta^T}{\ktau}\right).$$
	
	In addition, by \Cref{lem:tau_min_index}, we have $\ktau \geq \delta^T \cdot \left(1-\frac{1}{\beta}\right) $.
	Thus, we find that indeed, by \Cref{eq:upperboundk}
	\begin{align*}
	\ell^T(\ktau) & 
	\leq \frac{\delta^T- V_2^{\ktau}}{\ktau} 
	\leq \frac{\delta^T - \beta\cdot \left(\delta^T - \ktau - \ktau\log\,\frac{\delta^T}{\ktau}\right)}{\ktau} 
	= (1 - \beta)\frac{\delta^T}{\ktau} + \beta + \beta\log\,\frac{\delta^T}{\ktau} 
	\leq \beta\log\,\frac{\beta}{\beta - 1}.\qedhere 
	\end{align*}
\end{proof}

\boundsmalltau*

\begin{proof}
	For ease of notation, in this lemma we will let $\Delta=\delta^T$.
	We will consider two cases and show the bound holds for both cases.
	
	\paragraph{Case 1: $\bm{\delta^A/\beta \leq \ktau \leq \delta^A}$.} 
	By \Cref{lem:tech_lem_all} with $\ktau^* =  \delta^A \geq \ktau$, we have
	\begin{align*}
	V_2^k & \geq \beta \cdot \left(\Delta - \delta^A - \ktau\log\frac{\delta^T}{\delta^A}\right).
	\end{align*}

	As a consequence, by \Cref{eq:upperboundk}, we have
	\begin{align*}
	\ell^T(\ktau) &
	\leq \frac{\D - V_2^k}{\ktau} \\
	& \leq \left(\D - \beta(\D - \delta^A) + \beta\ktau\log\frac{\D}{\delta^A}\right) / \ktau \\
	& = (1 - \beta)\frac{\D}{\ktau} + \beta\frac{\delta^A}{\ktau} + \beta\log\frac{\D}{\delta^A}\\
	&= \frac{\delta^A}{\ktau}((1 - \beta)\frac{\D}{\delta^A} + \beta) + \beta\log\frac{\D}{\delta^A} \\
	& \leq \beta ((1-\beta)\frac{\D}{\delta^A} + \beta) + \beta\log\frac{\D}{\delta^A} \\
	& \leq \beta^2 - \beta + \beta\log\frac{1}{\beta - 1},
	\end{align*}
	where the third inequality above holds because $\frac{\delta^A}{\ktau} \leq \beta$ and $\frac{\Delta}{\delta^A} \leq \frac{\Delta}{\ktau} \leq \beta/(\beta - 1)$, by \Cref{lem:tau_min_index} and the last inequality holds because $\beta ((1-\beta)\frac{\D}{\delta^A} + \beta) + \beta\log\frac{\D}{\delta^A}$ is maximized when $\frac{\Delta}{\delta^A} = 1/(\beta - 1)$ (as can be verified by differentiating with respect to $x=\frac{\D}{\delta^A}$).

	\paragraph{Case 2: $\bm{\ktau \leq \delta^A/\beta}$.}	
	Note that after the arrival of vertex $u$, the color load is at most $\beta$, by \Cref{obs:load_arrival}. 
	We may safely assume that $A \geq \beta \ktau$, since we can always increase $A$ to $\beta\ktau$ without increasing volume in $V_2^k$ (which we aim to lower bound), by \Cref{obs:load_arrival}.
	\begin{align}
	V_2^k &= \sum_{i = \ktau + 1}^{\D} \ell^{\D}(i) \nonumber\\
	&=\sum_{i = \ktau + 1}^{\D}\ell^{A}(i) +  \sum_{i = \ktau + 1}^{\D} (\ell^{\delta^A}(i) -  \ell^{A}(i)) +  \sum_{i = \ktau + 1}^{\D} (\ell^{\D}(i) -  \ell^{\delta^A}(i)) \nonumber \\
	&\geq (A - \beta\ktau) + \sum_{j = A + 1}^{\delta^A}\beta\cdot\frac{\delta^j- \ktau}{\delta^j} + \beta \cdot \left(\D - \delta^A - \ktau\log\frac{\D}{\delta^A}\right) \nonumber\\
	&\geq (A - \beta\ktau) + (\delta^A - A)\cdot \beta \cdot \frac{\delta^A - \ktau}{\delta^A} + \beta\cdot \left(\D - \delta^A - \ktau\log\frac{\D}{\delta^A}\right)  \label{eq:substitute}\\
	&\geq (\delta^A - \beta \ktau) \cdot \beta \cdot \frac{\delta^A - \ktau}{\delta^A} + \beta\cdot (\D - \delta^A) - \beta\ktau\log\frac{\D}{\delta^A}. \nonumber
	\end{align}
	The first inequality holds by \Cref{obs:load_arrival}, \Cref{cor:tech_lem} and \Cref{lem:tech_lem_all} with $\ktau^*=\delta^A \geq \beta \ktau\geq \ktau$. The second inequality holds since for $j > A$, $\delta^j \geq \delta^A$. For the last inequality, substituting $A$ with $\beta\ktau$, a lower bound of $A$, will only decrease \Cref{eq:substitute}, since the coefficient of $A$ is non-negative; i.e. $1 - \beta + \beta\frac{\ktau}{\delta^A} \geq 1 - \beta + \beta\frac{\ktau}{\D} \geq 1 - \beta + \beta\cdot (1 - \frac{1}{\beta}) = 0$, where the last step follows by \Cref{lem:tau_min_index}. Consequently, by \Cref{eq:upperboundk}, we have
	\begin{align*}
	\ell^T(k) &\leq \frac{\D - V_2^k}{\ktau} \\
	& \leq \frac{\D - \left( (\delta^A - \beta \ktau) \cdot \beta \cdot \frac{\delta^A - \ktau}{\delta^A} + \beta\cdot (\D - \delta^A) - \beta\ktau\log\frac{\D}{\delta^A}\right)}{k} \\
	& \leq (1 - \beta)\frac{\Delta}{\ktau} + \beta^2 + \beta - \beta^2\frac{\ktau}{\delta^A}  + \beta\log\frac{\D}{\delta^A} \\
	& = (1 - \beta)\frac{\Delta}{\ktau} + \beta^2 + \beta - \beta^2\frac{\ktau}{\delta^A}  + \beta(\log\frac{\D}{\ktau} + \log\frac{\ktau}{\delta^A}) \\
	&= \beta^2 + \beta + (\beta\log\frac{\ktau}{\delta^A} - \beta^2\frac{\ktau}{\delta^A}) + (\beta\log\frac{\D}{\ktau} + (1 - \beta)\frac{\D}{\ktau})\\
	&\leq \beta^2 + \beta + (\beta\log\frac{1}{\beta} - \beta) + (\beta\log\frac{\beta}{\beta - 1} - \beta) \\
	&=\beta^2 - \beta + \beta\log\frac{1}{\beta - 1}.\qedhere
	\end{align*}
\end{proof}

Finally, we will need the following simple inequalities for our analysis.

\begin{restatable}{fact}{betabound}
	\label{cl:beta}
	For $\beta \in (1,2)$ we have
	$\beta \leq \generalUBBeta$, as well as
	$\biUBBeta \leq \generalUBBeta$.
\end{restatable}

\begin{proof}
	
	For both inequalities, we rely on $x - 1\geq \log(x)$ for all $x\geq 1$ to obtain the claimed inequalities. For the first, we have
	\begin{align*}
	\generalUBBeta - \beta & = \beta^2 - \beta + \beta\log\frac{1}{\beta - 1} - \beta = \beta\left((\beta - 1) - 1 - \text{log}\,(\beta - 1)\right) \geq 0.
	\end{align*}
	
	For the second inequality, we have
	\begin{align*}
	\beta^2 - \beta + \beta\log\frac{1}{\beta - 1}  -\beta\log\frac{\beta}{\beta - 1} & = \beta(\beta - 1 - \log\beta) \geq 0.\qedhere
	\end{align*}
\end{proof}

\subsection{Tight Example for Bounded Water Filling}
\label{sec:tight-example}
Here we give a tight instance for \Cref{alg:cap_water_filling} in general graphs, showing our analysis in \Cref{sec:upper_bound_bipartite} is tight. We use the same construction shown in \Cref{thm:lower-general}. Moreover, we assume that the state is ``old'' until phase $k = (\beta - 1)n$. We only consider the case when $b$ is sufficiently large. For sufficiently large $t$, where $t < k$, we can see the color status for vertex $v$ is roughly 
\[
\ell_{v}^{t}(i)=\left\{
\begin{matrix}
\beta\log\frac{\beta}{\beta - 1} & i \leq (1-\frac{1}{\beta})t\\
\sum_{x = i}^{t}\frac{\beta}{x} & (1-\frac{1}{\beta})t < i \leq k. 
\end{matrix}
\right.
\]
Meanwhile, the color status for vertex $u_t$ is roughly
\[
\ell_{u_t}^{t}(i)=\left\{
\begin{matrix}
0 & i \leq (1-\frac{1}{\beta})t\\
\beta & (1-\frac{1}{\beta})t < i \leq k.
\end{matrix}
\right.
\]
After the arrival of $u_k$, the coin is up, and the final color status for $u_k$(in round $n$) is 
\[
\ell_{u_k}^{n}(i)=\left\{
\begin{matrix}
\beta\log\frac{1}{\beta - 1} & i \leq \frac{(\beta - 1)^2}{\beta}n\\
\beta^2 -\beta + \beta\log\frac{1}{\beta - 1}n & \frac{(\beta - 1)^2}{\beta} < i < (\beta - 1)n\\
\sum_{x = i}^{n} \frac{\beta}{x} & (\beta - 1)n \leq i \leq n. 
\end{matrix}
\right.
\]
Consequently, our analysis for \Cref{alg:cap_water_filling} in \Cref{sec:upper_bound_bipartite} is tight and 1.777  is the best achievable competitive ratio for this algorithm.
	\section{Omitted Proofs of \Cref{sec:online-rounding}}\label{sec:unknown-delta-deferred}
Here we provide the missing proofs of lemmas and theorem deferred from \Cref{sec:online-rounding}, restated here for ease of reference.

We start by bounding the number of colors used during each \phase.
\unknownnumcolors*
\begin{proof}
	As $\Delta_i \geq (6 \log n) / p^3$, we have $\E[C_i] = \E[|\mathcal{S}_i|] \leq \alpha\Delta_i\cdot p \leq \alpha \cdot 6(\log n)/p^2$. Plugging $\epsilon=p$ into the upper multiplicative tail bound of \Cref{chernoff}, we get 
	\begin{align*}
		\Pr[C\geq \alpha\Delta_i\cdot p(1+p)] & \leq \exp\left(-\frac{\alpha\Delta_i\cdot p(1+p)}{3}\right) \leq \exp\left(-\frac{((6 \log n) / p^3)\cdot p^3}{3}\right) = 1/n^2.\qedhere
	\end{align*}
\end{proof}

The main technical lemma of this section, bounding the maximum degree of the uncolored graph $U_{i+1}$ in terms of its $i^{th}$ \phase counterpart, $U_i$, is as follows.

\unknowndegdecrease*

Before proving this lemma (in turn deferred to \Cref{sec:deg-decrease}), we show how it implies our main theorem, restated below.

\roundingunknown*
\begin{proof}
	For our proof, we will require the following fact.
	\begin{fact}\label{fact:exps}
		All $p \in [0,1/10]$ satisfy
		$(1-p-4p^2)\geq \exp(-2\cdot p)$ and $(1-p+7p^2) \leq \exp(-p/4)$.
	\end{fact}
	For $p = \sqrt[\rootconstant]{(24\log n)/\Delta'} \leq 1/10$ to hold, we need $\Delta'\geq \multiplicativeconstant \cdot \log n$. That is, $c=\multiplicativeconstant$.
	
	By \Cref{rand-max-deg-decrease} and \Cref{fact:exps}, $\Pr[\Delta_{i+1}\leq \Delta_i \cdot \exp(-2\cdot p)] \leq \Pr[\Delta_{i+1}\leq \Delta_i \cdot (1-p-4p^2)] \leq 3/n^2$, provided $\Delta_i\geq (24\log n)/p^4$. 
	By our choice of $p=\sqrt[\rootconstant]{(24\log n)/\Delta'}$, this implies that for all $i<\numphases$,
	$$\Delta\cdot \exp(-2\cdot p)^i \geq \Delta\cdot p^8 \geq \Delta'\cdot p^8 = (24 \log n)/p^4.$$
	Consequently, if we let $A_i \triangleq[\bigwedge_i \Delta_i \geq (24\log n)/p^4]$ be an indicator for the event that $\Delta_i$ is large enough to appeal to \Cref{rand-num-colors-used} and \Cref{rand-max-deg-decrease} for \phase $i$, then taking union bound (\Cref{ub}) over all $j < i$, we have 
	$$\Pr[\overline{A_i}] = \Pr\left[\Delta_i \leq (24\log n)/p^4\right]\leq \Pr\left[\bigvee_{j < i} (\Delta_{j} \leq \Delta\cdot \exp(-2\cdot p)^j)\right] \leq n\cdot 3/n^3 = 3/n^2.$$
	
	Now, by \Cref{rand-max-deg-decrease} and $p\leq 1/10$, we have 
	$$\Pr\left[\Delta_i - \Delta_{i+1} \leq p(1-7p) \cdot \Delta_i \;\middle|\; A_i\right] = \Pr\left[\Delta_{i+1}\geq \Delta_i\cdot (1-p+7p^2) \;\middle|\; A_i \right] \leq 6/n^2.$$
	
	On the other hand, by \Cref{rand-num-colors-used}, if we denote by $C_i$ the number of colors used during the $i^{th}$ \phase, then the probability of any of the $C_i$ being large is at most 
	\begin{align*}
		\Pr\left[C_i \geq \alpha\Delta_i\cdot p(1+p) \;\middle|\; A_i\right] \leq 1/n^2.
	\end{align*}
	
	Now, by $\alpha\in [1,2]$ and $p\leq 1/10$, we find that $\alpha+54p\leq \alpha(1+27p)\leq \alpha\cdot \frac{1+p}{1-7p}$. Therefore, if we let $B_i = \mathds{1}[C_i \geq (\alpha+54p)\cdot (\Delta_i - \Delta_{i+1})]$ be the bad event that we use a significantly higher number of colors in \phase $i$ than the amount by which we decrease the maximum degree in the uncolored graph in that \phase. Then, we have 
	\begin{align*}
		\Pr[B_i] & \leq \Pr[C_i \geq (\alpha+54p)\cdot (\Delta_i - \Delta_{i+1}) \mid A_i] + \Pr[\overline{A_i}] \\
		& \leq \Pr[C_i \geq \alpha\cdot \Delta_i \cdot p(1+p) \mid A_i ] + \Pr[\Delta_i - \Delta_{i+1} \leq p(1-7p)\cdot \Delta_i \mid A_i] + \Pr[\overline{A_i}] \\
		& \leq 1/n^2 + 6/n^2 + 3/n^2 = 10/n^2.
	\end{align*}
	
	Therefore, by union bound, we have that with probability at least $1-10/n$, the number of colors used during the \phases is at most 
	\begin{align*}
		\sum_i (\alpha + 54p) \cdot (\Delta_i - \Delta_{i+1}) & \leq (\alpha + 54p)\cdot \Delta.
	\end{align*}

	Finally, we upper bound the number of colors used by the greedy step of Line \ref{line:greedy-unknown}, by upper bounding the uncolored subgraph's maximum degree before \Cref{line:greedy-unknown}. We note that by \Cref{rand-max-deg-decrease} and \Cref{fact:exps}, we have $\Pr[\Delta_{i+1} \geq \Delta_i\cdot \exp(-p/4) \mid A] \leq \Pr[\Delta_{i+1} \geq \Delta_i \cdot (1-p+7p^2) \mid A] \leq 6/n^2$.
	Therefore, we find that the final uncolored subgraph $U$ has maximum degree $\Delta(U)\leq \Delta\cdot p$, as  
	\begin{align*}
	\Pr[\Delta(U)\geq \Delta\cdot p] & \leq \Pr[\Delta_{\numphases}\geq \Delta\cdot \exp(-p/4\cdot \numphases)] \\
	& \leq \Pr\left[\bigvee_i \left( \Delta_{i+1} \geq \Delta_i\cdot \exp(-p/4)\right) \right] \\
	& \leq 
	\Pr\left[\bigvee_i \left(\Delta_{i+1} \geq \Delta_i\cdot \exp(-p/4) \right)\, \Bigg\vert\, A \right] + \Pr[\overline{A}] \\
	& \leq n\cdot 6/n^2 + 3/n \\
	& = 9/n.
	\end{align*}
	Consequently, the greedy step of \Cref{line:greedy-unknown} uses a further $2\Delta\cdot p$ colors, and so \Cref{alg:round-unknown} is an $(\alpha+56p)$-competitive online edge coloring algorithm.
\end{proof}

	\subsection{Progress in degree decrease}\label{sec:deg-decrease}
	
	In this section we will show that each \phase $i$ of \Cref{alg:round-unknown} with $\Delta_i\geq 24(\log n)/p^3$ decreases the maximum degree of the uncolored graph by a $1/(1-p \pm O(p^2))$ factor. That is, we will prove \Cref{rand-max-deg-decrease}. As outlined in \Cref{sec:online-rounding}, our general approach will be to bound the number of times each near-maximum-degree vertex $v$ in $U_i$ is matched during the \phase and the number of times it is matched without having an edge colored. 
	
	For the remainder of this section, we will need the following random variables.
	First, for any vertex $v$ and index $i$, we let $d_i(v)$ denote $v$'s degree in the uncolored subgraphs $U_i$.
	Moreover, for each edge $e$ we let $L^{(i)}_{e,j} = x^{(i)}_j$ if $j\in \mathcal{S}_i$ and zero otherwise, and similarly $L^{(i)}_{v,j} \triangleq \sum_{e\ni v} L^{(i)}_{e,j}$.
	We refer to the above as the \emph{load} of edge $e$ and vertex $v$ in color $j$ of \phase $i$. Finally, we denote by $\ell^{(i)}_e \triangleq \sum_j L^{(i)}_{e,j}$ and $\ell^{(i)}_v \triangleq \sum_j L^{(i)}_{v,j}$ the load of the edge $e$ and vertex $v$ in the sampled colors of \phase $i$. Clearly, as each color index $j$ is in $\mathcal{S}_i$ with probability $p$, and as each edge is fractionally matched exactly once, we have that $\E[\ell^{(i)}_e] = p$ and therefore $\E[\ell^{(i)}_v] = d_i(v)\cdot p$. The following lemma asserts that these variables are concentrated around their mean. In all notation, we omit $i$, which will be clear from context.
	
	\begin{lem}\label{load-concentration}
		If $\Delta_i\geq (24\log n)/p^3$, then 
		\begin{enumerate}
			\item\label{load-edge} for each edge $e$ we have $\Pr[\ell_e \geq p(1+p)] \leq 1/n^4$, and
			\item\label{load-vertex} for each vertex $v$ of degree $d_i(v)\geq \Delta_i/2$ in $U_i$ we have $\Pr[|\ell_v - d_i(v)\cdot p| \geq d_i(v)\cdot p^2] \leq 2/n^3$.
		\end{enumerate} 
	\end{lem}
	\begin{proof}
		As noted above, $\E[\ell_e] = p$. Moreover, by the $(p^3/12\log n)$-boundedness of $f$ we have that $\ell_e=\sum_j L_{e,j}$ is the sum of bounded independent variables $L_{e,j} \in [0,p^3/12\log n]$. So, by Chernoff bounds (\Cref{chernoff}) with $\epsilon=p$, we  obtain
		\begin{align*}
			\Pr[\ell_e \geq p(1+p)] = \Pr\left[\ell_e \geq \E[\ell_e]\cdot (1+p)\right] \leq \exp\left(-\frac{p\cdot p^2}{3 p^3/(12 \log n)}\right) = \exp\left(-4\log n\right) = 1/n^4.
		\end{align*}
		
		Similarly, as noted above, $\E[\ell_v] = p\cdot d_i(v)$. Moreover, as $x^{(i)}$ is a feasible fractional matching, we have $|L_{v,j}|\leq 1$ for all $j$. So, by Chernoff bounds (\Cref{chernoff}), with $\epsilon=p$, we obtain 
		\begin{align*}
		\Pr[|\ell_v - \E[\ell_v]| \geq p^2\cdot d_i(v)] & = \Pr[|\sum_j L_{v,j} - \E[L_{v,j}]| \geq p\cdot  \sum_j \E[L_{v,j}]] \\
		& \leq 2\exp\left(-\frac{d_i(v)\cdot p\cdot p^2}{3}\right) \\
		& \leq 2\exp\left(-\frac{\Delta_i \cdot p\cdot p^2}{6}\right) \\
		& \leq 2\exp\left(-3\log n\right) \\
		& \leq 2/n^3.\qedhere
		\end{align*}		
	\end{proof}
	
	We will now want to bound the number of times a vertex is matched during a \phase. We will rely on \Cref{load-concentration} together with the following lemma.
	\begin{lem}\label{per-edge-p^3-bounded}
		Let $\vec{x}$ be a fractional matching with $\max_e x_e \leq p^4/(12 \log n)$. Then for each edge $e$, \marking run with input $\vec{x}$ outputs a matching $\mathcal{M}$ which matches each edge $e$ with probability 
		$$x_e \cdot (1-3p) \leq \Pr[e\in \mathcal{M}]\leq x_e$$
	\end{lem}
	\begin{proof}	
		The upper bound on $\Pr[e\in \mathcal{M}]$ is true for all $\vec{x}$. For the lower bound, we have that by \Cref{peredge-general}, as $p\in [0,1/10]$ and as we may safely assume $n\geq 2$ (otherwise the problem is trivial), we have
		\begin{align*}
		\Pr[e\in \mathcal{M}] &\geq x_e \cdot (1 - 11p\sqrt[3]{p\cdot \log (12\log n/p^3)/12\log n})  \\
		& \geq x_e \cdot (1 - 11p\sqrt[3]{3p\log (1/p)/12\log n + p\cdot \log (12\log n)/12\log n}) \\
		& \geq x_e \cdot (1 - 11p\sqrt[3]{3(1/e)/ 12\log n + p}) & p\in [0,1] \\		
		& \geq x_e \cdot (1-11p\sqrt[3]{3/(e\cdot 12\log 2) + p}) & n\geq 2 \\
		& \geq x_e \cdot (1-11p\sqrt[3]{3/(e\cdot 12\log 2) + 1/10}) & p\leq 1/10 \\		
		& \geq x_e\cdot (1-3p).\qedhere
		\end{align*} 
	\end{proof}
	
	Relying on \Cref{load-concentration}.\ref{load-vertex} and \Cref{per-edge-p^3-bounded}, we obtain the following bounds on $M_v$, the number of times $v$ is matched during the $i^{th}$ \phase.
	
	\begin{lem}\label{match-bound}
		If $\Delta_i \geq (24\log n)/p^4$, for each vertex $v$ with degree at least $d_i(v)\geq \Delta_i/2$, then $M_v,$ the number of times $v$ is matched during the $i^{th}$ \phase, satisfies
		\begin{enumerate}
			\item $\Pr[M_v \geq d_i(v) \cdot p(1+4p)]\leq 3/n^4$.
			\item $\Pr[M_v \leq d_i(v) \cdot p(1-5p)]\leq 3/n^3$.
	\end{enumerate}	
	\end{lem}
	\begin{proof}
		Let $M^{j}_v$ be an indicator variable for the event that $v$ is matched in $M_{i,j}$. 
		For any instantiation of the variables $L_{e,j}$,
		\Cref{per-edge-p^3-bounded} implies that each edge $e$ is matched in $M_{i,j}$ with probability $L_{e,j}\cdot (1-3p)\leq \Pr[e\in M_{i,j}] \leq L_{e,j}$, and so by linearity we have 
		$L_{v,j}\cdot (1-3p) \leq \Pr[M^j_{v}] \leq L_{v,j}$. 
		In particular, if we let $A\triangleq [d_i(v)\cdot p(1-p) \leq \ell_v\leq d_i(v)\cdot p(1+p)]$, then, by linearity we have both $\E[M_v \mid A]\leq d_i(v)\cdot p(1+p)$ as well as $\E[M_v \mid A] \geq d_i(v)\cdot p(1-p)(1-3p) \geq d_i(v)\cdot p(1-4p)$.
		Now, clearly, $M_v = \sum_{j\in \mathcal{S}_i} M^j_v$ is the sum of binary random variables.
		Moreover, for any subset $\mathcal{S}_i$ sampled, these $\{M^j_v \mid j\in \mathcal{S}_i\}$ are independent, as all matchings $M_{i,j}$ for $j\in \mathcal{S}_i$ are computed using independent copies of \marking. By Chernoff's upper tail bound (\Cref{chernoff}) with $\epsilon = 2p$, we thus obtain
		\begin{align*}
			\Pr[M_v \geq d_i(v) \cdot p(1+4p) \mid A]& \leq \Pr[M_v \geq d_i(v)\cdot p(1+p)(1+2p) \mid A] & \\
			& \leq \Pr[M_v \geq \E[M_v \mid A] \cdot (1+2p) \mid A] & 
			\\ 			
			& \leq \exp\left(-\frac{\E[M_v \mid A] \cdot 4p^2}{3}\right) & 
			\\
			& \leq \exp\left(-\frac{d_i(v)\cdot p(1-4p) \cdot 4p^2}{3}\right) & 
			\\
			& \leq \exp\left(-\frac{(48\log n)p^3(1-4p)}{3p^4}\right) & \\
			& \leq \exp\left(-4\log n\right) & p\leq 1/5\\
			& \leq 1/n^4.
		\end{align*}
		Therefore, we obtain the first claim, as
		\begin{align*}
		\Pr[M_v \geq d_i(v)\cdot p(1+4p)] & \leq \Pr[M_v \geq d_i(v) \cdot p(1+4p)\mid A] + \Pr[\overline{A}] \leq 3/n^3.
		\end{align*}

		Similarly, by Chernoff's lower tail bound (\Cref{chernoff}) with $\epsilon = p$, we obtain 
		\begin{align*}
		\Pr[M_v \leq d_i(v) \cdot p(1-5p) \mid A]& \leq \Pr[M_v \leq d_i(v)\cdot p(1-p)(1-3p)(1-p) \mid A] & \\
		& \leq \Pr[M_v \leq \E[M_v \mid A] \cdot (1-p) \mid A] & 
		\\ 			
		& \leq \exp\left(-\frac{\E[M_v \mid A] \cdot p^2}{2}\right) &
		\\
		& \leq \exp\left(-\frac{d_i(v)\cdot p(1-p)(1-3p) \cdot p^2}{2}\right) & 
		\\
		& \leq \exp\left(-\frac{12(\log n)p^3(1-p)(1-3p)}{2p^4}\right) & \\
		& \leq \exp\left(-3\log n\right) & p\leq 1/10\\
		& \leq 1/n^3.
		\end{align*}
		From the above we obtain the second claim, as
		\begin{align*}
		\Pr[M_v \leq d_i(v)\cdot p(1-5p)] & \leq \Pr[M_v \leq d_i(v) \cdot p(1-5p)\mid A] + \Pr[\overline{A}] \leq 3/n^3.\qedhere
		\end{align*}	
	\end{proof}

	The above lemma asserts that the number of times a vertex $v$ of high degree in $U_i$ is matched during the $i^{th}$ \phase is $\Theta(d_i(v)\cdot p)$. The following lemma relies on the theory of Negative Association (NA, see \Cref{sec:na}) to show that all but $O(d_i(v)\cdot p^2)$ matches of $v$ during this \phase result in an edge of $v$ being colored.

	\begin{lem}\label{num-recolors}
		If $\Delta_i \geq (24\log n)/p^3$, for each vertex $v$ with degree at least $d_i(v)\geq \Delta_i/2$, the number of times $v$ is matched along a previously colored edge, $R_v$, satisfies $\Pr[R_v \geq 2d_i(v)\cdot p^2] \leq 2/n^2$.
	\end{lem}
	\begin{proof}
		Fix the realizations of $L_{e,j}$ for all $e,j$. For any edge $e\ni v$, let $M_{e,j}\triangleq \mathds{1}[e\in M_{i,j}]$ be an indicator for edge $e$ being matched in \iteration $j$ of \phase $i$. By the 0-1 rule, since at most one edge $e\ni v$ is in any matching, for each $j$ the binary variables $\{M_{e,j} \mid e\ni v\}$ are NA. On the other hand, for $j\neq j'$ the joint distributions $\{M_{e,j} \mid e\ni v\}$ and $\{M_{e,j'} \mid e\ni v\}$ are independent. Thus, by closure of NA distributions under independent union (Property \ref{prop:ind-union}), the $\{M_{e,j} \mid j\in \mathcal{S}_i,e\ni v\}$ are NA. By closure of NA distributions under monotone increasing functions of disjoint variables (Property \ref{prop:monotone-fns}), if we let $R_e \triangleq \sum_j M_{e,j} \cdot \min\{1,\sum_{j'<j} M_{e,j'}\}$ denote the number of times $e$ is matched and not colored, then these $\{R_e \mid e\ni v\}$ are NA. In this terminology, we have that $R_v = \sum_{e\ni v} R_e$ is the sum of NA variables. Moreover, as the $M_{e,j}$ are NA and as $\E[M_{e,j}]\leq L_{e,j}$ by \Cref{peredge-general}, we have by the definition of NA variables (see \eqref{eq:NA-def}) that
		\begin{align*}
		\E\left[\sum_j M_{e,j}\cdot \sum_{j'<j} M_{e,j'}\right] & \leq \sum_j \E\left[M_{e,j}\right]\cdot \E\left[\sum_{j'<j} M_{e,j'}\right] \leq \sum_j L_{e,j}\cdot \sum_{j'<j} L_{e,j'} \leq \ell_e \cdot \ell_e.
		\end{align*} 
		
		Let $A=\mathds{1}[\forall e\ni v:\,\ell_e \leq p(1+p)]$ be an indicator for the high probability event that every edge $e\ni v$ has load at most $2p$ in the sampled matchings. 
		\begin{align*}
		\E[R_e \mid A] & \leq \E\left[\sum_j M_{e,j}\cdot \sum_{j'<j} M_{e,j'} \,\Bigg\vert \, A\right]\leq \E[\ell_e \mid A] \cdot \E[\ell_e \mid A] 
		\leq p^2(1+p)^2.
		\end{align*}
		Therefore, by linearity of expectation, $\E[R_v] = \sum_{e\ni v} \E[R_e] \leq d_i(v)\cdot p^2(1+p)^2$.
		Now, as $d_i(v)\geq \Delta_i/2 \geq 12(\log n)/p^3$ and as $R_v=\sum_e R_e$ is the sum of binary NA variables, we can upper bound $R_v$ using the upper multiplicative Chernoff bound of \Cref{lem:nuod-chernoff} with $\epsilon = \sqrt{p}$ to obtain 
		\begin{align*}
		\Pr[R_v \geq d_i(v)\cdot p^2(1+p)^2(1+\sqrt{p}) \mid A] & \leq \exp\left(-\frac{ d_i(v)\cdot p^2(1+p)^2\cdot p}{3}\right)\leq  \exp\left(-\frac{12\log n}{3}\right) \leq \frac{1}{n^2}.
		\end{align*}
		Observing that for $p\leq 1/10$ we have $2\leq (1+p)^2(1+\sqrt{p})$, we find that 
		\begin{align*}
		\Pr[R_v\geq 2d_i(v)\cdot p^2 \mid A] & \leq \Pr[R_v \geq d_i(v)\cdot p^2(1+p)^2(1+\sqrt{p}) \mid A]\leq 1/n^2.
		\end{align*}
		
		Now, by \Cref{load-concentration}.\ref{load-edge} we have for every $e\ni v$ that $\Pr[\ell_e \geq p(1+p)] \leq 1/n^3$ and so by union bound we have $\Pr[\overline{A}] \leq n\cdot 1/n^3 = 1/n^2$. 
		We therefore conclude that indeed
		\begin{align*}
			\Pr[R_v \geq 2\cdot d_i(v)\cdot p^2] & \leq \Pr[R_v \geq 2\cdot d_i(v)\cdot p^2 \mid A] + \Pr[\overline{A}] \leq 2/n^2.\qedhere 
		\end{align*}
	\end{proof}

	\Cref{rand-max-deg-decrease}, restated below for ease of reference, follows from lemmas \ref{match-bound} and \ref{num-recolors} and union bound of relevant subsets of vertices.
	\unknowndegdecrease*
	\begin{proof}
		For each vertex $v$, the decrease in $v$'s degree in the uncolored subgraph during the $i^{th}$ \phase, denoted by $D_v \triangleq d_{i}(v) - d_{i+1}(v)$, is precisely the number of times $v$ is matched and its matched edge is colored. That is, in the terminology of \Cref{match-bound} and \Cref{num-recolors}, $D_v = M_v - R_v$. 
		So, by \Cref{match-bound}, every maximum degree vertex $v$ in $U_i$ (i.e. $d_i(v)=\Delta_i\geq \Delta_i/2$)
		satisfies
		\begin{align*}
			\Pr[d_{i+1}(v) \leq \Delta_i\cdot (1-p-3p^2)] & = \Pr[d_{i+1}(v) \leq d_i(v)\cdot (1-p-3p^2)] \\
			& = \Pr[d_{i}(v) - d_{i+1}(v) \geq d_i(v)\cdot p(1+3p)] \\
			& = \Pr[D_v \geq d_i(v)\cdot p(1+3p)] \\ 
			& \leq \Pr[M_v \geq d_i(v)\cdot p(1+3p)] \\
			& \leq 3/n^4.
		\end{align*}		
		
		The first claim then follows by union bound over all maximum degree vertices $v$ in $U_i$.
		\begin{align*}
			\Pr[\Delta_{i+1} &\leq \Delta_i\cdot (1-p-3p^2)] \leq \sum_{v:\, d_i(v)=\Delta_i} \Pr[d_{i+1}(v) \leq \Delta_i\cdot (1-p-3p^2)] \leq 3/n^3.
		\end{align*}
				
		Now, we let $\lambda \triangleq p(1-7p)$ and note that $(1-\lambda)\cdot \Delta_i \geq \Delta_i/2$, since $p\leq 1/2$. 
		All vertices $v$ of degree $d_i(v)\leq (1-\lambda) \cdot \Delta_i$ in $U_i$ clearly have $d_{i+1}(v) \leq d_i(v)\leq (1-\lambda) \cdot \Delta_i$. On the other hand, for every $v$ with $d_i(v)\geq (1-\lambda)\cdot \Delta_i \geq \Delta_i/2$, we have by lemmas \ref{match-bound} and \ref{num-recolors} that
		\begin{align*}
			\Pr[d_{i+1}(v) \geq (1-\lambda)\cdot \Delta_i] & \leq
			\Pr[d_{i+1}(v) \geq (1-\lambda)\cdot d_i(v)] \\
			& = \Pr[d_{i}(v) - d_{i+1}(v) \leq d_i(v)\cdot \lambda] \\
			& = \Pr[D_v \leq d_i(v)\cdot \lambda] \\
			& = \Pr[D_v \leq d_i(v)\cdot p(1-7p)]  \\
			& \leq \Pr[M_v \leq d_i(v)\cdot p(1-5p)] + \Pr[R_v \geq d_i(v)\cdot p\cdot 2p] \\
			& \leq 6/n^3.
		\end{align*}
		
		The second claim then follows by union bound over all vertices $v$ of degree $d_i(v)\geq (1-\lambda)\cdot \Delta_i$ in $U_i$, recalling that $\lambda = p(1-7p)$, as
		\begin{align*}
			\Pr[\Delta_{i+1} \geq (1-\lambda)\cdot \Delta_i] & \leq \sum_{v:\, d_i(v)\geq (1-\lambda)\cdot \Delta_i}\Pr[d_{i+1}(v) \geq (1-\lambda)\cdot \Delta_i] \leq 6/n^2.\qedhere
		\end{align*}
	\end{proof}
	\vspace{-0.1cm}
\section{Improved $o(1)$ Terms for Known $\Delta$}\label{sec:known-delta}

In this section we present an improved algorithm for known $\Delta$. Recall that for the known $\D$ regime, using our online rounding scheme of \Cref{rounding-unknown} we obtained a $1+O(\sqrt[12]{(\log n)/\Delta})$-competitive algorithm (see \Cref{known-delta-from-rounding}).
In this section we will show how to decrease this competitive ratio to $1+O(\sqrt[4]{(\log n)/\Delta})$. That is, we improve the $o(1)$ term when $\Delta=\omega(\log n)$.

We now turn to describing our approach, starting with an offline description. Iterating over $c\in [\Delta]$, we compute and color a matching $M_c$ in the uncolored subgraph $G \setminus \bigcup_{c'=1}^{c-1} M_{c'}$.
We then color the remaining uncolored subgraph with new colors using the greedy algorithm. 
This approach can be implemented online, by
iteratively running online matching algorithms on the relevant uncolored subgraphs to compute and color matchings. More concretely, when a vertex $v$ arrives, we iterate over $c\in [\Delta]$ and update $M_c$ in the current uncolored graph $G\setminus \bigcup_{c'=1}^{c-1} M_{c'}$, as follows. We run the next step of the online matching algorithm used to compute $M_c$ in the current uncolored graph after $v$'s arrival in this subgraph. We then color $v$'s newly-matched edge (if any) using color $c$. Finally, we run steps of the greedy algorithm on the remaining uncolored edges of $v$. 

For our analysis, we will analyze the above algorithm according to its offline description.
Since the greedy algorithm requires a number of colors linear in its input graph's maximum degree, our objective will be to reduce the uncolored subgraph's maximum degree to $o(\Delta)$ w.h.p.~after computing and coloring the first $\Delta$ matchings. In particular, this will require us to match each maximum-degree vertex in $G$ with probability roughly one for each of these $\Delta$ matchings.
One way of matching vertices $v$ of degree $\Delta$ in the uncolored subgraph with probability roughly one is to guarantee each edge $e\ni v$ a probability of roughly $\frac{1}{\Delta}$ of being matched. An online matching algorithm which does just this is obtained from \Cref{peredge-general} applied to the trivial fractional matching which assigns a value of $\frac{1}{\Delta}$ to each edge.
We will refer by $\marking_d$ to the application of \textsc{marking} to the trivially-feasible fractional matching assigning $x_e=\epsilon=\frac{1}{d}$ for each edge in a graph of (known) maximum degree at most $d$.

\begin{restatable}[]{cor}{peredge}
	\label{peredge}
	Algorithm $\marking_d$ is an online matching algorithm which in graphs of maximum degree at most $d$ outputs a matching $\mathcal{M}$ which matches each edge $e$ with probability
	$$\frac{1}{d}\cdot  \Big(1-11\sqrt[3]{(\log d)/d}\Big) \leq \Pr[e\in \mathcal{M}] \leq \frac{1}{d}.$$
\end{restatable}

The first natural approach given \Cref{peredge} is to iteratively run $\marking_\Delta$. However, as shown in \Cref{sec:bad-example-marking}, this approach is suboptimal.
Instead, we will increase the probability of high-degree vertices in the uncolored subgraph to have an edge colored, by running $\marking_d$ with a tighter upper bound $d$ than $\Delta$ for the uncolored graph's maximum degree for each phase. Unfortunately, upon arrival of some vertex $v$, we do not know the uncolored graph's maximum degree for all phases, as this depends on \emph{future} arrivals and random choices of our algorithm. To obtain a tight (up to $o(\Delta)$) bound $d$ on the uncolored graph's maximum degree for each phase, we divide the $\Delta$ coloring \iterations into \emph{\phases} of $\ell = \sqrt{\Delta \log n}$ \iterations each, during which we use the same upper bound. As $\ell = o(\Delta)$ and $\ell = \omega(\log n)$, this gives us sharply concentrated upper bounds $d_{i+1}$ on the resulting uncolored graph's maximum degree at the end of each \phase $i$, which in turn serves as a tight upper bound for the next \phase. 
This results in the desired rate of decrease in the uncolored graph's maximum degree, namely $1-o(1)$ per \iteration. Greedy thus runs on a subgraph of maximum degree $o(\Delta)$. Our $1+o(1)$ competitive ratio follows.

\vspace{-0.1cm}
\subsection{The Improved Algorithm}\label{sec:ec-algo}
We now present our online edge coloring algorithm, starting with an offline description.
Our algorithm consists of $\Delta$ \iterations, equally divided into $\sqrt{\Delta/\log n}$ \phases. During 
each \iteration of \phase $i$, we color a matching output by $\marking_{d_i}$ run on $U_i$ -- the uncolored subgraph prior to \phase $i$, 
for $d_i \triangleq  \Delta - i\cdot (\ell - 8\sqrt{\ell \log n})$. After all \phases, we run greedy with new colors, starting with $\Delta+1$. In the online implementation, after each online vertex $v$'s arrival, 
for \phase $i=1,2,\dots$, we run the next step of $\ell=\sqrt{\Delta \log n}$ independent runs of $\marking_{d_i}$ in $U_i$, color newly-matched edges and update $U_{i'}$ for $i'>i$ accordingly. We then greedily color $v$'s remaining uncolored edges with new colors.
The algorithm's pesudocode is given in \Cref{alg:online-ec-known-delta}.

\begin{algorithm}[H] 
	\caption{Improved Randomized Edge Coloring for Known $\Delta$}
	\label{alg:online-ec-known-delta}
	
	\begin{algorithmic}[1]
		\Require Online bipartite graph $G(L,R,E)$ with maximum degree $\Delta = \omega(\log n)$.
		
		\Ensure Integral $(1+o(1))\Delta$ edge coloring, w.h.p.

		\State let $\ell \triangleq  \lfloor \sqrt{\Delta \log n}\rfloor$.		\Comment{\phase length}	
		\State let $d_i \triangleq  \Delta - i\cdot (\ell - 8\sqrt{\ell \log n})$ for $i\in [0,ֿ\Delta/\ell]$. \Comment{degree upper bound for each \phase}
		\State for all $i$, denote by $U_i$ the \textbf{online} subgraph of $G$ not colored by colors $[i\cdot \ell]$.						
		\For{each arrival of a vertex $v\in R$}
		\For{\phase $i=0,1,\dots,\lfloor \Delta/\ell\rfloor -1$} \label{line:loop-start}	
		\For{colors $c\in [i\cdot \ell + 1, (i+1)\cdot \ell]$} 
		
		\State $M_{c}\leftarrow $ output of copy $c$ of $\marking_{d_i}$ on \textbf{current} $U_i$. \Comment{run next step of $\marking_{d_i}$}
		\If{an edge $e\in M_{c}$ is previously uncolored} \Comment{note: $e\ni v$}
		\State color $e$ using color $c$.
		\EndIf
		\EndFor 
		\EndFor
		\State run greedy on all uncolored edges of $v$, using new colors starting from $\Delta+1$. \label{line:greedy}
		\EndFor
	\end{algorithmic}
\end{algorithm}

\subsection{Analysis}\label{sec:ec-analysis}

The crux of our analysis is that for each \phase $i$, we have $d_i\geq \Delta(U_i)$ w.h.p. Consequently, the final uncolored subgraph after the $\Delta$ \iterations (and colors) has maximum degree at most $d_{\lfloor \Delta/\ell \rfloor} = o(\Delta)$, so greedily coloring this subgraph requires a further $o(\Delta)$ colors.
The following lemma
asserts that if $d_i\geq \Delta(U_i)$, then $d_{i+1}\geq \Delta(U_{i+1})$, w.h.p.

\begin{restatable}{lem}{degknown}\label{max-deg-bound-known-delta}
	For all $i\in [0,\Delta/\ell-1]$, if $\Delta(U_i) \leq d_i$, then $\Pr[\Delta(U_{i+1}) > d_{i+1}]\leq 1/n^3$.
\end{restatable}
\begin{proof}
	If $\Delta(U_i) \leq d_{i} - \ell$, the claim is trivial, as then $d_{i+1} \geq d_i - \ell \geq \Delta(U_{i}) \geq \Delta(U_{i+1})$. We therefore focus on the case $d_i - \ell \leq \Delta(U_i) \leq d_i$. For this latter case, we will rely on the fact that for all $i\leq \Delta/\ell$, we have 
	$d_i = \Delta - i\cdot (\ell - 8\sqrt{\ell \log n}) \geq (\Delta/\ell)\cdot 8\sqrt{\ell \log n} > 3\Delta^{3/4}\log^{1/4} n.$

	Vertices of degree less than $\Delta(U_i) - \ell \leq d_i - \ell < d_{i+1}$ clearly have degree at most $d_{i+1}$ in $U_{i+1}$, as we only decrease their degree in the uncolored subgraph over time. We therefore turn our attention to vertices $v$ of degree at least $\Delta(U_i) - \ell$ in $U_i$. Such a vertex $v$ cannot have more than $\ell$ edges colored during \phase $i$, regardless of our algorithm's random choices, as at most one of $v$'s edges is colored per \iteration. So, before each color $c$ used in the \phase, $v$ has at least $\Delta(U_i) - 2\ell \geq d_i - 3\ell$ uncolored edges, each of which is matched by $\marking_{d_i}$ with probability at least $\frac{1}{{d_i}}(1-11\sqrt[3]{(\log d_i)/d_i})$. Therefore, if we let $X_c\triangleq \mathds{1}[\bigvee_{e\ni v} $e$ \textrm{ colored } c]$ be an indicator for the event that $v$ has an edge colored $c$, then, regardless of the realization $\vec{x}$ of variables $X_{c-1},X_{c-2},\dots,X_{(i-1)\cdot \ell + 1}$ corresponding to previous \iterations of the $i^{th}$ \phase, vertex $v$ will have an edge colored $c$ with probability at least \vspace{-0.15cm}
	\begin{align*}
	\Pr[X_c=1 \mid (X_{c-1},X_{c-2},\dots,X_{(i-1)\cdot \ell + 1}) = \vec{x}] & \geq (d_i - 3\ell) \cdot \frac{1}{d_i}\left(1-11\sqrt[3]{(\log d)/d}\right) \\
	& \geq 1-\frac{3\ell}{d_i}-11\sqrt[3]{(\log d_i)/d_i} \\
	& \geq 1 - \left(\frac{3}{8}+\frac{11}{2}\right)\sqrt[4]{(\log n)/\Delta} \\
	& \geq 1 - 6\sqrt[4]{(\log n)/\Delta},
	\end{align*}
	where the penultimate inequality follows from $\ell \leq \sqrt{\Delta \log n}$ and $n\geq d_i\geq 8\Delta^{3/4} \log^{1/4} n$.
	
	Therefore, the expected decrease of $v$'s degree in the uncolored graph during the $\ell$ \iterations of the  $i^{th}$ \phase is at least $\E[\sum_{c=i\cdot \ell + 1}^{(i+1)\cdot \ell} X_c] \geq \ell\cdot \left(1 - 6\sqrt[4]{(\log n)/\Delta}\right) \geq \ell - 6\sqrt{\ell \log n}$. But the probability of $v$ having an edge colored $c$ is at least $1 - 6\sqrt[4]{(\log n)/\Delta}$ 	
	\emph{independently} of previous colors during the \phase. Consequently, we can appeal to standard coupling arguments (\Cref{coupling}) together with Hoeffding's inequality (\Cref{hoeffding}) to show that the sum of these $\ell$ binary variables satisfies
	\begin{align*} 
	\Pr\left[\sum_{c=i\cdot \ell + 1}^{(i+1)\cdot \ell} X_c \leq \ell - 6\sqrt{\ell \log n} - \sqrt{2\ell \log n}\right] & \leq \exp\left(-\frac{2(\sqrt{2\ell \log n})^2}{\ell}\right) = 1/n^4.
	\end{align*}
	
	Put otherwise, $v$'s degree in the uncolored subgraph decreases during \phase $i$ by less than $\ell - 6\sqrt{\ell \log n} - \sqrt{2\ell\cdot \log n} > \ell - 8\sqrt{\ell \log n}$ with probability at most $1/n^4$.
	Thus, as $v$ has degree at most $\Delta(U_i)$ in $U_i$ by definition, we find that vertex $v$'s degree in $U_{i+1}$, denoted by $D_v$, satisfies 
	$$\Pr[D_v \geq d_{i+1}] = \Pr[D_v \geq d_i - \ell + 8\sqrt{\ell \log n}] \leq \Pr[D_v \geq \Delta(U_i) - \ell + 8\sqrt{\ell \log n}] \leq 1/n^4.$$ 
	Taking union bound over all vertices, the lemma follows.
\end{proof}

The above lemma implies this section's main result, given by the following theorem.

\begin{restatable}{thm}{knowndelta}\label{known-delta}
	\Cref{alg:online-ec-known-delta} is $(1+O(\sqrt[4]{(\log n)/\Delta}))$-competitive w.h.p.~in $n$-vertex bipartite graphs with known maximum degree $\Delta=\Omega(\log n)$.
\end{restatable}

\begin{proof}
	\Cref{alg:online-ec-known-delta} computes a feasible edge coloring. It colors each edge, by  \Cref{line:greedy}, and each color class -- computed during \iterations or by greedy -- constitutes a matching (here we rely on the colors used by greedy and the \phases being disjoint). It remains to bound the number of colors this algorithm uses. Each \phase requires at most $\ell$ colors, so the \phases require at most $\Delta$ colors. The number of colors the greedy step requires is at most twice the maximum degree of the remaining uncolored subgraph after the \phases, which we now bound.
	
	Let $A_i \triangleq \mathds{1}[\Delta(U_i)\leq  d_i]$ be an indicator for the event that $d_i$ upper bounds $\Delta(U_i)$. By \Cref{max-deg-bound-known-delta} we have that $\Pr[\overline{A_i} \mid A_{i-1}, A_{i-2},\dots ] = \Pr[\overline{A_i} \mid A_{i-1}] \leq 1/n^3$. Also, trivially $\Pr[\overline{A_0}] = 0$. Taking union bound over all $i$, we find that the probability of any $A_i$ not being one is at most
	\begin{align*}
	\Pr\left[\bigvee_{i=0}^{\Delta/\ell} \overline{A_i}\right] & \leq 	\sum_{i=1}^{\Delta/\ell} \Pr\left[\overline{A_i}, A_{i-1}, A_{i-2},\dots \right] \leq 	\sum_{i=1}^{\Delta/\ell} \Pr\left[\overline{A_i} \mid A_{i-1}, A_{i-2},\dots \right] \leq (\Delta/\ell)/n^3 \leq 1/n^2.
	\end{align*}
	
	Consequently, all applications of $\marking_{d_{i+1}}$ during \phase $i+1$ match each edge of $U_{i+1}$ with probability at least $\frac{1}{d_{i+1}}(1-11\sqrt[3]{(\log d_{i+1})/d_{i+1}})$, as required by our analysis for  \phase $i+1$.
	Moreover, $d_i\geq \Delta(U_i)$ for all $i\in [0,\Delta/\ell]$ w.h.p.~implies that the uncolored subgraph following the $\Delta/\ell$ \phases has maximum degree at most 
	\begin{align*}
	d_{\lfloor\Delta/\ell\rfloor} & = \Delta - \lfloor\Delta/\ell\rfloor\cdot (\ell - 8\sqrt{\ell \log n}) \\
	& \leq \ell + (\Delta/\ell)\cdot 8\sqrt{\ell \log n} \\
	& \leq \Delta^{1/2} \log^{1/2} n + 8\Delta^{3/4}\log^{1/4} n \\
	& \leq 9\Delta^{3/4}\log^{1/4} n.
	\end{align*}
	The greedy algorithm therefore colors the remaining uncolored graph using at most a further $18\Delta^{3/4}\log^{1/4} n - 1$ colors. That is, it uses $\Delta\cdot O(\sqrt[4]{(\log n)/\Delta})$ colors in the range $\Delta+1,\Delta+2,\dots$. The theorem follows, including the stated bound for $\Delta = \omega(\log n)$.
\end{proof}

\noindent\textbf{Remark.} \Cref{alg:online-ec-known-delta} is $(1+O(\sqrt[4]{(\log n)/\Delta}))$ competitive w.h.p.~for all $\Delta=\Omega(\log n)$ large enough, so for $\Delta=\Omega(\log n)$ large enough it yields a constant competitive ratio strictly smaller than $2$.
	\section{Omitted Proofs of \Cref{sec:hardness}}
\label{sec:hardness-omitted}

Here we present our proof for the lower bound for online edge coloring in general graphs.
\lowergeneral*
\begin{proof}
	The adversarial instance has $m + 1$ possible futures. Neither the number of phases $m$ nor the choice of future are known to the online algorithm. There is a state associated with the input, and there are two possible states, ``old'' and ``new''. Initially, the graph contains $m!$ vertices and the state is ``old''. There are $m'\leq m$ phases in total. We use $V_k$ to denote the set of online vertices which arrive in phase $k$ ($k \in [m']$) and $V_0$ to denote the initial $m!$ vertices which arrive in phase $0$. Moreover, we use $v^{k}_{i}$ to denote the $i^{th}$ vertex arrived in phase $k$. In phase $k$, newly-arrived vertices have degree $k$. If the state is ``old'', $m!/k$ vertices arrive and the $i^{th}$ vertex, $v^{k}_{i}$, is adjacent to $v^{0}_{i}, v^{0}_{m!/k + i}, \cdots, v^{0}_{(k-1)m!/k + i}$. On the other hand, if the state is ``new''  and it changed from ``old'' to ``new'' at the end of phase $t$ ($k > t$), then $m!/kt$ vertices arrive and the $i^{th}$ vertex, $v^{k}_{i}$, will neighbor $v^{t}_{i}, v^{t}_{m!/kt + i}, \cdots v^{t}_{m!(k-1)/kt + i}$. At the end of phase $k$, the adversary decide whether to switch state to ``new''. Notice that the state can only transition from ``old'' to ``new".

	Again, we let $x_{kj}$ denote the average assignment of color $j$ to edges of phase $k$, but this time only if the state is ``old'' during this phase.
	The following constraints still hold for the same reason as Constraints \eqref{eq:bipartite-1} and \eqref{eq:bipartite-2} for the bipartite hard instance of \Cref{thm: lower-bipartite}.
	\begin{equation}
	\label{general-volume}
	\sum_{j = 1}^{k}x_{k,j} \geq 1 \quad \forall k.
	\end{equation}
	\begin{equation}
	\label{general-capacity-old}
	\sum_{k = j}^{m}x_{i,j} \leq \alpha \quad \forall j.
	\end{equation}
	Furthermore, We use $y_{k, j}^{t}$ to denote the average assignment of color $j$ to edges between $V_k$ and $V_t$ when the state transitions from ``old'' to ``new'' in phase $t$. (I.e., this is the average assignment of color $j$ to edges of phase $k>t$, for $t$ the phase at which the transition occurred.)
	
	Again, as each edge between a $V_t$ vertex and its neighbor in $V_k$ ($k>t$) must be fractionally colored, we have  
	\begin{equation}
	\sum_{j = 1}^{k}y_{k, j}^{t} \geq 1 \quad \forall t < k \leq m. \label{general-capacity-new1}
	\end{equation}
	Moreover, the maximum load of every vertex for every color is at most $\alpha$, and so we have
	\begin{align}
	t\cdot x_{t, j} + \sum_{k = t + 1}^{m}y_{k, j}^{t} &\leq \alpha \quad \forall 1 \leq j\leq t\leq m \label{general-capacity-new2}\\ 
	\sum_{k = j}^{m}y_{k, j}^{t} &\leq \alpha \quad \forall 1\leq t < j \leq m \label{general-capacity-new3}\\ 
	k\cdot y_{k, j}^{t} &\leq \alpha \quad \forall 1 \leq t < k\leq m. \label{general-capacity-new4} 
	\end{align}
	To summarize, constraints \eqref{general-volume}-\eqref{general-capacity-new4} for any $m$ are all satisfied by any $\alpha$-competitive online fractional edge coloring algorithm on this distribution of inputs. Therefore, the optimal value of an LP with objective of minimizing $\alpha$ subject to these constraints is a lower bound on the optimal competitive ratio $\alpha$ of any such online algorithm on general graphs.
	Using commercial solvers, we solve this LP for $m = 50$ and find that its optimal value, which lower bounds any algorithm's competitive ratio on general graphs, is $\generalLB$. Again, using the same trick as \Cref{sec:dense-graph}, we find that this lower bound also holds for dense graphs.
\end{proof}
	\section{Extension to Multigraphs}\label{sec:multigraphs}
In this section we outline the extension of our positive results to multigraphs (the negative results carry over trivially).

\paragraph{Fractional Algorithms.} Our fractional results carry over unchanged to multigraphs. To see 
this, note that our algorithms' analyses do not require the graph to be simple, as our analysis 
implies a bound on the maximum load after each edge has its value increased, and the relevant bounds 
do not require there to be no parallel edges.

\paragraph{Randomized Algorithms.}
For multigraphs, we ``merge'' parallel edges into a single edge. When running \marking on some fractional matching, we have a merged edge's fractional assignment be the sum of its constituent edges' fractional assignment. If each edge has multiplicity a sufficiently small $o(\Delta)$ term, this would assign each edge a value of $o(1)$. By the properties of \marking (\Cref{peredge-general}), this implies that when we round a fractional matching $\vec{x}$ to compute a matching $\mathcal{M}$, each edge $e$ is matched in $\mathcal{M}$ with probability $x_e\cdot (1-o(1))\leq \Pr[e\in \mathcal{M}] \leq x_e$. Our arguments carry through, though with possibly worse $o(1)$ terms.

We note that the above stipulation that each edge have bounded 
multiplicity is necessary in order to obtain $(1+o(1))$ competitiveness for known $\Delta$.\footnote{\citet{aggarwal2003switch} showed that one cannot even achieve $5/4-\epsilon$ competitiveness in multigraphs with unbounded edge multiplicities, though under the possibly harder adversarial \emph{edge arrival} model.} 
\begin{restatable}{obs}{multigraphs}\label{cor:multigraphs}
	No algorithm is $(1+o(1))$ competitive on multigraphs of arbitrary multiplicity.
\end{restatable}
\begin{proof}	
	By \cite{cohen2018randomized}, no online matching algorithm outputs a matching of expected size 
	$c\cdot n$ in $2$-regular $2n$-vertex bipartite graphs under one-sided arrivals, for some constant $c<1$.
	Given an input online $2$-regular graph, we simulate the online arrival of a multigraph with $k$ copies of each edge of the input (simple) graph, to obtain a $2k$-regular multigraph with each edge having multiplicity $k$. 
	Given a $(1+\epsilon)$-competitive edge coloring algorithm for multigraphs of maximum degree $\Delta = 2k$,
	we can randomly pick one of the $(1+\epsilon)\cdot 2k$ color classes upon initialization and output that matching. 
	For $2$-regular graphs on $2n$ vertices, which have $2\cdot n$ edges, this results in a matching in the multigraph (which has $2kn$ edges) of expected size $\frac{2kn}{(1+\epsilon)\cdot 2k} = n/(1+\epsilon)$, from which we 
	conclude $\epsilon = \Omega(1)$.
\end{proof}
	\section{Useful Probabilistic Inequalities}\label{sec:prob-basics}

For completeness, we cite here some useful probabilistic inequalities and notions of negative dependence, starting with the latter.

\subsection{Negative Association and Other Negative Dependence Properties.}\label{sec:na} In our analysis we rely on several notions of negative dependence between random variables. In particular, one notion we will rely on is the notion of \emph{negative association}, introduced by \citet{khursheed1981positive} and \citet{joag1983negative}.
\begin{Def}[Negative Association \cite{khursheed1981positive,joag1983negative}]\label{eq:NA-def}
	A joint distribution $X_1,X_2,\dots,X_n$ is said to be \emph{negatively associated (NA)} if for any two functions $f,g$ both monotone increasing or both monotone decreasing, with $f(\vec{X})$ and $g(\vec{X})$ depending on disjoint subsets of the $X_i$, $f(\vec{X})$ and $g(\vec{X})$ are negatively correlated; i.e., 
	$$\E[f(\vec{X}) \cdot g(\vec{X})] \leq \E[f(\vec{X})]\cdot \E[g(\vec{X})].$$
\end{Def}

Clearly, independent random variables are NA. Another class of NA distributions is captured by the \emph{zero-one rule}. 
This rule asserts that if $X_1,X_2,\dots,X_n$ are zero-one random variables whose sum is always at most one, $\sum_i X_i\leq 1$, then $X_1,X_2,\dots,X_n$ are NA (see \cite{dubhashi1996balls}).
Additional, more complex, NA distributions can be ``built'' from simpler NA distributions using the following closure properties.

\begin{enumerate}[(P1)]
	\item \label{prop:ind-union} \underline{Independent Union.} If $X_1,X_2,\dots,X_n$ are NA, $Y_1,Y_2,\dots,Y_m$ are NA, and $\{X_i\}_i$ are independent of $\{Y_j\}_j$, then $X_1,X_2,\dots,X_n,Y_1,Y_2,\dots,Y_m$ are NA.
	\item \label{prop:monotone-fns} \underline{Concordant monotone functions.} Let $f_1,f_2,\dots,f_k:\mathbb{R}^n\rightarrow \mathbb{R}$ be functions, all monotone increasing or all monotone decreasing, with the $f_i(\vec{X})$ depending on disjoint subsets of the $\{X_i\}_i$. Then, if $X_1,X_2,\dots,X_n$ are NA, so are $f_1(\vec{X}),f_2(\vec{X}),\dots,f_k(\vec{X})$.
\end{enumerate}

Negative association implies several useful properties, including the applicability of Chernoff-Hoeffding type bounds~\cite{dubhashi1996balls} (we elaborate on this below). In addition, NA clearly implies  pairwise negative correlation. More generally, NA implies the stronger notion of \emph{negative orthant dependence}. 

\begin{Def}
	A joint distribution $X_1,X_2,\dots,X_n$ is said to be \emph{Negative Upper Orthant Dependent (NUOD)}, if for all $\vec{x}\in \mathbb{R}^{n}$ it holds that
	$$\Pr[\bigwedge_{i\in [n]} X_i\geq x_i] \leq \prod_{i\in [n]} \Pr[X_i \geq x_i],$$ 
	and \emph{Negative Lower Orthant Dependent (NLOD)} if for all $\vec{x}\in \mathbb{R}^n$ it holds that
	$$\Pr[\bigwedge_{i\in [n]} X_i\leq x_i] \leq \prod_{i\in [n]} \Pr[X_i \leq x_i].$$ 
	A joint distribution is said to be \emph{Negative Orthant Dependent (NOD)} if it is both NUOD and NLOD.
\end{Def}

\begin{lem}[NA variables are NOD (\cite{joag1983negative,dubhashi1996balls})]\label{lem:marginal-NA} If $X_1,\dots,X_n$ are NA, then they are NOD.
\end{lem}

In our analysis we will prove some scaled Bernoulli random variables are NUOD.  To motivate our interest in this form of negative dependence, we note that for binary NUOD variables $X_1,X_2,\dots,X_n$, we have that for each set $I\subseteq [n]$, $\Pr[\bigwedge_{i\in I} X_i = 1] \leq \prod_{i\in I}\Pr[X_i = 1]$. 
As shown by \citet[proof of Theorem 3.2, with $\lambda=1$]{panconesi1997randomized}, 
this property implies that the moment generating function of the sum of the $X_i$ is upper bounded by the moment generating function of the sum of \emph{independent} copies of the $X_i$ variables. A simple extension of their argument shows the same holds if the $X_i$ are NUOD \emph{scaled} Bernoulli variables. As in \cite{panconesi1997randomized}, following the standard proofs of Chernoff-Hoeffding type bounds, this upper bound on the moment generating function implies the applicability of the following upper tail bounds to the sum of NUOD  scaled Bernoulli variables ``as though these variables were independent''.

\begin{lem}[Chernoff Bound for NUOD Bernoulli Variables, \cite{panconesi1997randomized}]\label{lem:nuod-chernoff}
	Let $X=\sum_i X_i$ be the sum of binary NUOD random variables $X_1, X_2,...,X_n$. Then, for any $\eps\in [0,1]$ and $R\geq \E[X]$
	\begin{align*}
	\Pr[X > (1+\delta)\cdot R] & 
	\leq \exp\left(\frac{-\eps\cdot R}{3}\right).
	\end{align*}
\end{lem}

\begin{lem}[Bernstein's Inequality for NUOD Scaled Bernoulli Variables]\label{lem:bernstein-nuod}
	Let X be the sum of NUOD random variables $X_1, X_2,...,X_n$ with $X_i\in \{0,M_i\}$ and $M_i\leq M$ for each $i\in [n]$.
	Then, if $\sigma^2=\sum_{i=1}^n  Var(X_i)$, we have for all $a>0$, 
	$$\Pr[X > \E[X] + a] \leq \exp\left(\frac{-a^2}{2(\sigma^2 + aM/3) }\right).$$
\end{lem}

In addition we will use the following simple coupling argument, stated here for completeness.
\begin{lem}\label{coupling}
	Let $X_1,X_2,\dots,X_m$ be random variables and $Y_1,Y_2,\dots,Y_m$ be binary random variables such that $Y_i=f_i(X_1,X_2,\dots,X_i)$ for all $i$ such that for all $\vec{x}\in \mathbb{R} ^m$, 
	$$\Pr\left[Y_i = 1 \,\Bigg\vert\, \bigwedge_{\ell\in [i]} X_\ell = x_\ell\right] \leq p_i.$$
	Then, if $Z_i=Bernoulli(p_i)$ are independent random variables, we have
	$$\Pr\left[\sum_i Y_i \geq k\right] \leq \Pr\left[\sum_i Z_i \geq k\right].$$
\end{lem}

\subsection{Standard Concentration Inequalities}
In addition, we will need the following standard tail bounds.

\begin{lem}[Chernoff Bound]\label{chernoff}
	Let $X_1,X_2,\dots,X_m$ be independent random variables. 
	If 
	$X_i\in [0,b]$ always, 
	then, for all $\epsilon\in (0,1)$ and $R\geq \E[X]$, we have
	$$\Pr[X\geq (1+\epsilon) \cdot R] \leq \exp\left(-\frac{R\cdot \epsilon^2}{3b}\right),$$
	$$\Pr[X\leq (1-\epsilon) \cdot \E[X]] \leq \exp\left(-\frac{\E[X]\cdot \epsilon^2}{2b}\right).$$		
\end{lem}

\begin{lem}
	[Hoeffding's Inequality]\label{hoeffding}
	Let $X_1,X_2,\dots,X_m$ be independent random variables. 
	If 
	$X_i\in [a_i,b_i]$ always, 
	then, for all $\delta$, 
	$$\Pr[X\geq \E[X] + \delta] \leq \exp\left(-\frac{2\delta^2}{\sum_i (b_i-a_i)^2}\right),$$	
	$$\Pr[X\leq \E[X] - \delta] \leq \exp\left(-\frac{2\delta^2}{\sum_i (b_i-a_i)^2}\right).$$		
\end{lem}

\subsection{Conditional Union Bound}
Finally, we will also need the following simple extension of union bound, proven below for completeness.

\begin{lem}\label{ub}
	Let $A_1,A_2,\dots,A_n$ be random indicator variables such that $\Pr[\overline{A_i} \mid \bigwedge_{j<i} A_{j}] \leq p$. Then $$\Pr\left[\bigvee_i \overline{A_i}\right] \leq n\cdot p.$$
\end{lem}
\begin{proof}
	Let $B_i = \overline{A_i} \wedge \bigwedge_{j<i} A_j$ be the event that $i$ is the first index $j$ for which $\overline{A_j}$ holds. Then 
	\begin{align*}
		\Pr\left[\bigvee_i \overline{A_i}\right] & = \Pr\left[\bigvee_i \overline{B_i}\right] 
		 \leq \sum_i \Pr\left[\overline{B_i}\right],
	\end{align*}
	by standard union bound. But then, we find that the above is at most
	\begin{align*}
		\sum_i \Pr\left[\overline{A_i} \wedge \bigwedge_{j<i} A_j\right]
		& = \sum_i \Pr\left[\overline{A_i} \,\middle\vert\, \bigwedge_{j<i} A_j\right] \cdot \Pr\left[\bigwedge_{j<i} A_j\right] 
		 \leq \sum_i \Pr\left[\overline{A_i} \mid A_{i-1}\right] 
		 \leq n\cdot p.\qedhere
	\end{align*}
\end{proof}
	\bibliographystyle{acmsmall}
	\bibliography{refs}

\begin{thebibliography}{}

\bibitem[\protect\citeauthoryear{Aggarwal, Motwani, Shah, and Zhu}{Aggarwal
  et~al\mbox{.}}{2003}]{aggarwal2003switch}
{\sc Aggarwal, G.}, {\sc Motwani, R.}, {\sc Shah, D.}, {\sc and} {\sc Zhu, A.}
  2003.
\newblock Switch scheduling via randomized edge coloring.
\newblock In {\em Proceedings of the 44th Symposium on Foundations of Computer
  Science (FOCS)}. 502--512.

\bibitem[\protect\citeauthoryear{Alon}{Alon}{2003}]{alon2003simple}
{\sc Alon, N.} 2003.
\newblock A simple algorithm for edge-coloring bipartite multigraphs.
\newblock {\em Information Processing Letters (IPL)\/}~{\em 85,\/}~6, 301--302.

\bibitem[\protect\citeauthoryear{Azar, Cohen, and Roytman}{Azar
  et~al\mbox{.}}{2017}]{azar2017online}
{\sc Azar, Y.}, {\sc Cohen, I.~R.}, {\sc and} {\sc Roytman, A.} 2017.
\newblock Online lower bounds via duality.
\newblock In {\em Proceedings of the 28th Annual ACM-SIAM Symposium on Discrete
  Algorithms (SODA)}. 1038--1050.

\bibitem[\protect\citeauthoryear{Bahmani, Mehta, and Motwani}{Bahmani
  et~al\mbox{.}}{2012}]{bahmani2012online}
{\sc Bahmani, B.}, {\sc Mehta, A.}, {\sc and} {\sc Motwani, R.} 2012.
\newblock Online graph edge-coloring in the random-order arrival model.
\newblock {\em Theory of Computing (TOC)\/}~{\em 8,\/}~1, 567--595.

\bibitem[\protect\citeauthoryear{Bar-Noy, Motwani, and Naor}{Bar-Noy
  et~al\mbox{.}}{1992}]{bar1992greedy}
{\sc Bar-Noy, A.}, {\sc Motwani, R.}, {\sc and} {\sc Naor, J.~S.} 1992.
\newblock The greedy algorithm is optimal for on-line edge coloring.
\newblock {\em Information Processing Letters (IPL)\/}~{\em 44,\/}~5, 251--253.

\bibitem[\protect\citeauthoryear{Berger and Rompel}{Berger and
  Rompel}{1991}]{berger1991simulating}
{\sc Berger, B.} {\sc and} {\sc Rompel, J.} 1991.
\newblock Simulating $(\log^cn)$-wise independence in {NC}.
\newblock {\em Journal of the ACM (JACM)\/}~{\em 38,\/}~4, 1026--1046.

\bibitem[\protect\citeauthoryear{Bhattacharya, Chakrabarty, Henzinger, and
  Nanongkai}{Bhattacharya et~al\mbox{.}}{2018}]{bhattacharya2018dynamic}
{\sc Bhattacharya, S.}, {\sc Chakrabarty, D.}, {\sc Henzinger, M.}, {\sc and}
  {\sc Nanongkai, D.} 2018.
\newblock Dynamic algorithms for graph coloring.
\newblock In {\em Proceedings of the 29th Annual ACM-SIAM Symposium on Discrete
  Algorithms (SODA)}. 1--20.

\bibitem[\protect\citeauthoryear{Birnbaum and Mathieu}{Birnbaum and
  Mathieu}{2008}]{birnbaum2008line}
{\sc Birnbaum, B.} {\sc and} {\sc Mathieu, C.} 2008.
\newblock On-line bipartite matching made simple.
\newblock {\em ACM SIGACT News\/}~{\em 39,\/}~1, 80--87.

\bibitem[\protect\citeauthoryear{Buchbinder, Jain, and Naor}{Buchbinder
  et~al\mbox{.}}{2007}]{buchbinder2007online}
{\sc Buchbinder, N.}, {\sc Jain, K.}, {\sc and} {\sc Naor, J.~S.} 2007.
\newblock Online primal-dual algorithms for maximizing ad-auctions revenue.
\newblock In {\em Proceedings of the 15th Annual European Symposium on
  Algorithms (ESA)}. 253--264.

\bibitem[\protect\citeauthoryear{Chang, Li, and Pettie}{Chang
  et~al\mbox{.}}{2018}]{chang2018optimal}
{\sc Chang, Y.-J.}, {\sc Li, W.}, {\sc and} {\sc Pettie, S.} 2018.
\newblock An optimal distributed (${\Delta}$+ 1)-coloring algorithm?
\newblock In {\em Proceedings of the 50th Annual ACM Symposium on Theory of
  Computing (STOC)}. 445--456.

\bibitem[\protect\citeauthoryear{Cohen and Wajc}{Cohen and
  Wajc}{2018}]{cohen2018randomized}
{\sc Cohen, I.~R.} {\sc and} {\sc Wajc, D.} 2018.
\newblock Randomized online matching in regular graphs.
\newblock In {\em Proceedings of the 29th Annual ACM-SIAM Symposium on Discrete
  Algorithms (SODA)}. 960--979.

\bibitem[\protect\citeauthoryear{Cole, Ost, and Schirra}{Cole
  et~al\mbox{.}}{2001}]{cole2001edge}
{\sc Cole, R.}, {\sc Ost, K.}, {\sc and} {\sc Schirra, S.} 2001.
\newblock Edge-coloring bipartite multigraphs in ${O}({E} \log {D})$ time.
\newblock {\em Combinatorica\/}~{\em 21,\/}~1, 5--12.

\bibitem[\protect\citeauthoryear{Devanur, Jain, and Kleinberg}{Devanur
  et~al\mbox{.}}{2013}]{devanur2013randomized}
{\sc Devanur, N.~R.}, {\sc Jain, K.}, {\sc and} {\sc Kleinberg, R.~D.} 2013.
\newblock Randomized primal-dual analysis of ranking for online bipartite
  matching.
\newblock In {\em Proceedings of the 24th Annual ACM-SIAM Symposium on Discrete
  Algorithms (SODA)}. 101--107.

\bibitem[\protect\citeauthoryear{Duan, He, and Zhang}{Duan
  et~al\mbox{.}}{2019}]{duan2019dynamic}
{\sc Duan, R.}, {\sc He, H.}, {\sc and} {\sc Zhang, T.} 2019.
\newblock Dynamic edge coloring with improved approximation.
\newblock In {\em Proceedings of the 30th Annual ACM-SIAM Symposium on Discrete
  Algorithms (SODA)}. To Appear.

\bibitem[\protect\citeauthoryear{Dubhashi, Grable, and Panconesi}{Dubhashi
  et~al\mbox{.}}{1998}]{dubhashi1998near}
{\sc Dubhashi, D.}, {\sc Grable, D.~A.}, {\sc and} {\sc Panconesi, A.} 1998.
\newblock Near-optimal, distributed edge colouring via the nibble method.
\newblock {\em Theoretical Computer Science (TCS)\/}~{\em 203,\/}~2, 225--252.

\bibitem[\protect\citeauthoryear{Dubhashi and Ranjan}{Dubhashi and
  Ranjan}{1996}]{dubhashi1996balls}
{\sc Dubhashi, D.} {\sc and} {\sc Ranjan, D.} 1996.
\newblock Balls and bins: A study in negative dependence.
\newblock {\em BRICS Report Series\/}~{\em 3,\/}~25.

\bibitem[\protect\citeauthoryear{Eden, Feldman, Fiat, and Segal}{Eden
  et~al\mbox{.}}{2018}]{eden2018economic}
{\sc Eden, A.}, {\sc Feldman, M.}, {\sc Fiat, A.}, {\sc and} {\sc Segal, K.}
  2018.
\newblock An economic-based analysis of ranking for online bipartite matching.
\newblock {\em arXiv preprint arXiv:1804.06637\/}.

\bibitem[\protect\citeauthoryear{Ehmsen, Favrholdt, Kohrt, and Mihai}{Ehmsen
  et~al\mbox{.}}{2010}]{ehmsen2010comparing}
{\sc Ehmsen, M.~R.}, {\sc Favrholdt, L.~M.}, {\sc Kohrt, J.~S.}, {\sc and} {\sc
  Mihai, R.} 2010.
\newblock Comparing first-fit and next-fit for online edge coloring.
\newblock {\em Theoretical Computer Science (TCS)\/}~{\em 411,\/}~16-18,
  1734--1741.

\bibitem[\protect\citeauthoryear{Elkin, Pettie, and Su}{Elkin
  et~al\mbox{.}}{2015}]{elkin20152delta}
{\sc Elkin, M.}, {\sc Pettie, S.}, {\sc and} {\sc Su, H.-H.} 2015.
\newblock (2${\Delta}$-1)-edge-coloring is much easier than maximal matching in
  the distributed setting.
\newblock In {\em Proceedings of the 26th Annual ACM-SIAM Symposium on Discrete
  Algorithms (SODA)}. 355--370.

\bibitem[\protect\citeauthoryear{Favrholdt and Mikkelsen}{Favrholdt and
  Mikkelsen}{2014}]{favrholdt2014online}
{\sc Favrholdt, L.~M.} {\sc and} {\sc Mikkelsen, J.~W.} 2014.
\newblock Online dual edge coloring of paths and trees.
\newblock In {\em Proceedings of the 12th Workshop on Approximation and Online
  Algorithms (WAOA)}. 181--192.

\bibitem[\protect\citeauthoryear{Favrholdt and Nielsen}{Favrholdt and
  Nielsen}{2003}]{favrholdt2003line}
{\sc Favrholdt, M.} {\sc and} {\sc Nielsen, N.} 2003.
\newblock On-line edge-coloring with a fixed number of colors.
\newblock {\em Algorithmica\/}~{\em 35,\/}~2, 176--191.

\bibitem[\protect\citeauthoryear{Feige}{Feige}{2018}]{feige2018tighter}
{\sc Feige, U.} 2018.
\newblock Tighter bounds for online bipartite matching.
\newblock {\em arXiv preprint arXiv:1812.11774\/}.

\bibitem[\protect\citeauthoryear{Feldman, Mehta, Mirrokni, and
  Muthukrishnan}{Feldman et~al\mbox{.}}{2009}]{feldman2009online}
{\sc Feldman, J.}, {\sc Mehta, A.}, {\sc Mirrokni, V.}, {\sc and} {\sc
  Muthukrishnan, S.} 2009.
\newblock Online stochastic matching: Beating 1-1/e.
\newblock In {\em Proceedings of the 50th Symposium on Foundations of Computer
  Science (FOCS)}. 117--126.

\bibitem[\protect\citeauthoryear{Fischer, Ghaffari, and Kuhn}{Fischer
  et~al\mbox{.}}{2017}]{fischer2017deterministic}
{\sc Fischer, M.}, {\sc Ghaffari, M.}, {\sc and} {\sc Kuhn, F.} 2017.
\newblock Deterministic distributed edge-coloring via hypergraph maximal
  matching.
\newblock In {\em Proceedings of the 58th Symposium on Foundations of Computer
  Science (FOCS)}. 180--191.

\bibitem[\protect\citeauthoryear{Gabow, Nishizeki, Kariv, Leven, and
  Osmau}{Gabow et~al\mbox{.}}{1985}]{gabow1985algorithms}
{\sc Gabow, H.~N.}, {\sc Nishizeki, T.}, {\sc Kariv, O.}, {\sc Leven, D.}, {\sc
  and} {\sc Osmau, T.} 1985.
\newblock Algorithms for edge-coloring graphs.
\newblock Tech. rep., Tohoku University.

\bibitem[\protect\citeauthoryear{Gamlath, Kapralov, Maggiori, Svensson, and
  Wajc}{Gamlath et~al\mbox{.}}{2019}]{gamlath2019online}
{\sc Gamlath, B.}, {\sc Kapralov, M.}, {\sc Maggiori, A.}, {\sc Svensson, O.},
  {\sc and} {\sc Wajc, D.} 2019.
\newblock Online matching with general arrivals.
\newblock {\em arXiv preprint arXiv:1904.08255\/}.

\bibitem[\protect\citeauthoryear{Gandham, Dawande, and Prakash}{Gandham
  et~al\mbox{.}}{2008}]{gandham2008link}
{\sc Gandham, S.}, {\sc Dawande, M.}, {\sc and} {\sc Prakash, R.} 2008.
\newblock Link scheduling in wireless sensor networks: Distributed
  edge-coloring revisited.
\newblock {\em Journal of Parallel and Distributed Computing\/}~{\em 68,\/}~8,
  1122--1134.

\bibitem[\protect\citeauthoryear{Ghaffari, Kuhn, Maus, and Uitto}{Ghaffari
  et~al\mbox{.}}{2018}]{ghaffari2018deterministic}
{\sc Ghaffari, M.}, {\sc Kuhn, F.}, {\sc Maus, Y.}, {\sc and} {\sc Uitto, J.}
  2018.
\newblock Deterministic distributed edge-coloring with fewer colors.
\newblock In {\em Proceedings of the 50th Annual ACM Symposium on Theory of
  Computing (STOC)}. 418--430.

\bibitem[\protect\citeauthoryear{Goel, Kapralov, and Khanna}{Goel
  et~al\mbox{.}}{2013}]{goel2013perfect}
{\sc Goel, A.}, {\sc Kapralov, M.}, {\sc and} {\sc Khanna, S.} 2013.
\newblock Perfect matchings in ${O}(n \log n)$ time in regular bipartite
  graphs.
\newblock {\em SIAM Journal on Computing (SICOMP)\/}~{\em 42,\/}~3, 1392--1404.

\bibitem[\protect\citeauthoryear{Goel and Mehta}{Goel and
  Mehta}{2008}]{goel2008online}
{\sc Goel, G.} {\sc and} {\sc Mehta, A.} 2008.
\newblock Online budgeted matching in random input models with applications to
  adwords.
\newblock In {\em Proceedings of the 19th Annual ACM-SIAM Symposium on Discrete
  Algorithms (SODA)}. 982--991.

\bibitem[\protect\citeauthoryear{Goldberg}{Goldberg}{1973}]{goldberg1973multigraphs}
{\sc Goldberg, M.~K.} 1973.
\newblock On multigraphs of almost maximal chromatic class.
\newblock {\em Diskret. Analiz\/}~{\em 23,\/}~3, 7.

\bibitem[\protect\citeauthoryear{Holyer}{Holyer}{1981}]{holyer1981np}
{\sc Holyer, I.} 1981.
\newblock The {NP}-completeness of edge-coloring.
\newblock {\em SIAM Journal on Computing (SICOMP)\/}~{\em 10,\/}~4, 718--720.

\bibitem[\protect\citeauthoryear{Huang, Kang, Tang, Wu, Zhang, and Zhu}{Huang
  et~al\mbox{.}}{2018}]{huang2018match}
{\sc Huang, Z.}, {\sc Kang, N.}, {\sc Tang, Z.~G.}, {\sc Wu, X.}, {\sc Zhang,
  Y.}, {\sc and} {\sc Zhu, X.} 2018.
\newblock How to match when all vertices arrive online.
\newblock In {\em Proceedings of the 50th Annual ACM Symposium on Theory of
  Computing (STOC)}. 17--29.

\bibitem[\protect\citeauthoryear{Huang, Peng, Tang, Tao, Wu, and Zhang}{Huang
  et~al\mbox{.}}{2019}]{huang2019tight}
{\sc Huang, Z.}, {\sc Peng, B.}, {\sc Tang, Z.~G.}, {\sc Tao, R.}, {\sc Wu,
  X.}, {\sc and} {\sc Zhang, Y.} 2019.
\newblock Tight competitive ratios of classic matching algorithms in the fully
  online model.
\newblock In {\em Proceedings of the 30th Annual ACM-SIAM Symposium on Discrete
  Algorithms (SODA)}. To Appear.

\bibitem[\protect\citeauthoryear{Joag-Dev and Proschan}{Joag-Dev and
  Proschan}{1983}]{joag1983negative}
{\sc Joag-Dev, K.} {\sc and} {\sc Proschan, F.} 1983.
\newblock Negative association of random variables with applications.
\newblock {\em The Annals of Statistics\/}, 286--295.

\bibitem[\protect\citeauthoryear{Kalyanasundaram and Pruhs}{Kalyanasundaram and
  Pruhs}{2000}]{kalyanasundaram2000optimal}
{\sc Kalyanasundaram, B.} {\sc and} {\sc Pruhs, K.~R.} 2000.
\newblock An optimal deterministic algorithm for online $b$-matching.
\newblock {\em Theoretical Computer Science (TCS)\/}~{\em 233,\/}~1, 319--325.

\bibitem[\protect\citeauthoryear{Karande, Mehta, and Tripathi}{Karande
  et~al\mbox{.}}{2011}]{karande2011online}
{\sc Karande, C.}, {\sc Mehta, A.}, {\sc and} {\sc Tripathi, P.} 2011.
\newblock Online bipartite matching with unknown distributions.
\newblock In {\em Proceedings of the 43rd Annual ACM Symposium on Theory of
  Computing (STOC)}. 587--596.

\bibitem[\protect\citeauthoryear{Karloff and Shmoys}{Karloff and
  Shmoys}{1987}]{karloff1987efficient}
{\sc Karloff, H.~J.} {\sc and} {\sc Shmoys, D.~B.} 1987.
\newblock Efficient parallel algorithms for edge coloring problems.
\newblock {\em J. Algorithms\/}~{\em 8,\/}~1, 39--52.

\bibitem[\protect\citeauthoryear{Karp, Vazirani, and Vazirani}{Karp
  et~al\mbox{.}}{1990}]{karp1990optimal}
{\sc Karp, R.~M.}, {\sc Vazirani, U.~V.}, {\sc and} {\sc Vazirani, V.~V.} 1990.
\newblock An optimal algorithm for on-line bipartite matching.
\newblock In {\em Proceedings of the 22nd Annual ACM Symposium on Theory of
  Computing (STOC)}. 352--358.

\bibitem[\protect\citeauthoryear{Khursheed and Lai~Saxena}{Khursheed and
  Lai~Saxena}{1981}]{khursheed1981positive}
{\sc Khursheed, A.} {\sc and} {\sc Lai~Saxena, K.} 1981.
\newblock Positive dependence in multivariate distributions.
\newblock {\em Communications in Statistics - Theory and Methods\/}~{\em
  10,\/}~12, 1183--1196.

\bibitem[\protect\citeauthoryear{Lev, Pippenger, and Valiant}{Lev
  et~al\mbox{.}}{1981}]{lev1981fast}
{\sc Lev, G.~F.}, {\sc Pippenger, N.}, {\sc and} {\sc Valiant, L.~G.} 1981.
\newblock A fast parallel algorithm for routing in permutation networks.
\newblock {\em IEEE transactions on Computers\/}~2, 93--100.

\bibitem[\protect\citeauthoryear{Lov{\'a}sz and Plummer}{Lov{\'a}sz and
  Plummer}{2009}]{lovasz2009matching}
{\sc Lov{\'a}sz, L.} {\sc and} {\sc Plummer, M.~D.} 2009.
\newblock {\em Matching theory}. Vol. 367.
\newblock American Mathematical Society.

\bibitem[\protect\citeauthoryear{Mahdian and Yan}{Mahdian and
  Yan}{2011}]{mahdian2011online}
{\sc Mahdian, M.} {\sc and} {\sc Yan, Q.} 2011.
\newblock Online bipartite matching with random arrivals: an approach based on
  strongly factor-revealing lps.
\newblock In {\em Proceedings of the 43rd Annual ACM Symposium on Theory of
  Computing (STOC)}. 597--606.

\bibitem[\protect\citeauthoryear{Mehta}{Mehta}{2013}]{mehta2013online}
{\sc Mehta, A.} 2013.
\newblock Online matching and ad allocation.
\newblock {\em Foundations and Trends{\textregistered} in Theoretical Computer
  Science\/}~{\em 8,\/}~4, 265--368.

\bibitem[\protect\citeauthoryear{Mikkelsen}{Mikkelsen}{2015}]{mikkelsen2015optimal}
{\sc Mikkelsen, J.~W.} 2015.
\newblock Optimal online edge coloring of planar graphs with advice.
\newblock In {\em International Conference on Algorithms and Complexity}.
  Springer, 352--364.

\bibitem[\protect\citeauthoryear{Mikkelsen}{Mikkelsen}{2016}]{mikkelsen2016randomization}
{\sc Mikkelsen, J.~W.} 2016.
\newblock Randomization can be as helpful as a glimpse of the future in online
  computation.
\newblock In {\em Proceedings of the 43rd International Colloquium on Automata,
  Languages and Programming (ICALP)}. 39.

\bibitem[\protect\citeauthoryear{Misra and Gries}{Misra and
  Gries}{1992}]{misra1992constructive}
{\sc Misra, J.} {\sc and} {\sc Gries, D.} 1992.
\newblock A constructive proof of vizing's theorem.
\newblock In {\em Information Processing Letters (IPL)}.

\bibitem[\protect\citeauthoryear{Motwani, Naor, and Naor}{Motwani
  et~al\mbox{.}}{1994}]{motwani1994probabilistic}
{\sc Motwani, R.}, {\sc Naor, J.~S.}, {\sc and} {\sc Naor, M.} 1994.
\newblock The probabilistic method yields deterministic parallel algorithms.
\newblock {\em Journal of Computer and System Sciences\/}~{\em 49,\/}~3,
  478--516.

\bibitem[\protect\citeauthoryear{Naor and Wajc}{Naor and
  Wajc}{2018}]{naor2018near}
{\sc Naor, J.~S.} {\sc and} {\sc Wajc, D.} 2018.
\newblock Near-optimum online ad allocation for targeted advertising.
\newblock {\em ACM Transactions on Economics and Computation (TEAC)\/}~{\em
  6,\/}~3-4, 16.

\bibitem[\protect\citeauthoryear{Panconesi and Srinivasan}{Panconesi and
  Srinivasan}{1997}]{panconesi1997randomized}
{\sc Panconesi, A.} {\sc and} {\sc Srinivasan, A.} 1997.
\newblock Randomized distributed edge coloring via an extension of the
  chernoff--hoeffding bounds.
\newblock {\em SIAM Journal on Computing (SICOMP)\/}~{\em 26,\/}~2, 350--368.

\bibitem[\protect\citeauthoryear{Petersen}{Petersen}{1898}]{petersen1898theoreme}
{\sc Petersen, J.} 1898.
\newblock Sur le th{\'e}oreme de tait.
\newblock {\em L'interm{\'e}diaire des Math{\'e}maticiens\/}~{\em 5}, 225--227.

\bibitem[\protect\citeauthoryear{Rasala and Wilfong}{Rasala and
  Wilfong}{2005}]{rasala2005strictly}
{\sc Rasala, A.} {\sc and} {\sc Wilfong, G.} 2005.
\newblock Strictly nonblocking wdm cross-connects.
\newblock {\em SIAM Journal on Computing (SICOMP)\/}~{\em 35,\/}~2, 449--485.

\bibitem[\protect\citeauthoryear{Seymour}{Seymour}{1979}]{seymour1979multi}
{\sc Seymour, P.~D.} 1979.
\newblock On multi-colourings of cubic graphs, and conjectures of fulkerson and
  tutte.
\newblock {\em Proceedings of the London Mathematical Society\/}~{\em 3,\/}~3,
  423--460.

\bibitem[\protect\citeauthoryear{Shannon}{Shannon}{1949}]{shannon1949theorem}
{\sc Shannon, C.~E.} 1949.
\newblock A theorem on coloring the lines of a network.
\newblock {\em Journal of Mathematics and Physics\/}~{\em 28,\/}~1-4, 148--152.

\bibitem[\protect\citeauthoryear{Tait}{Tait}{1880}]{tait1880remarks}
{\sc Tait, P.} 1880.
\newblock Remarks on the colourings of maps.
\newblock {\em Proc. R. Soc. Edinburgh\/}~{\em 10}, 729.

\bibitem[\protect\citeauthoryear{Tassiulas and Ephremides}{Tassiulas and
  Ephremides}{1992}]{tassiulas1992stability}
{\sc Tassiulas, L.} {\sc and} {\sc Ephremides, A.} 1992.
\newblock Stability properties of constrained queueing systems and scheduling
  policies for maximum throughput in multihop radio networks.
\newblock {\em IEEE transactions on automatic control\/}~{\em 37,\/}~12,
  1936--1948.

\bibitem[\protect\citeauthoryear{Vizing}{Vizing}{1964}]{vizing1964estimate}
{\sc Vizing, V.~G.} 1964.
\newblock On an estimate of the chromatic class of a p-graph.
\newblock {\em Diskret analiz\/}~{\em 3}, 25--30.

\bibitem[\protect\citeauthoryear{Wang and Wong}{Wang and
  Wong}{2015}]{wang2015two}
{\sc Wang, Y.} {\sc and} {\sc Wong, S. C.-w.} 2015.
\newblock Two-sided online bipartite matching and vertex cover: Beating the
  greedy algorithm.
\newblock In {\em Proceedings of the 42nd International Colloquium on Automata,
  Languages and Programming (ICALP)}. 1070--1081.

\end{thebibliography}
\end{document}